\documentclass[11pt,a4paper]{article}
\usepackage[a4paper]{geometry}
\usepackage{amssymb,latexsym,amsmath,amsfonts,amsthm}
\usepackage{graphicx}
\usepackage{epstopdf}
\usepackage{tikz}
\usepackage{pgflibraryshapes}
\usetikzlibrary{arrows,decorations.markings}
\usepackage{subfigure}
\usepackage{overpic}
\usepackage{fullpage}
\usepackage{comment,verbatim}
\usepackage{hyperref}
\usepackage[title,titletoc]{appendix}

\newcommand{\Ai}{\mathrm{Ai}}
\newcommand{\Bes}{\mathrm{Bes}}
\newcommand{\CHF}{\mathrm{CHF}}

\newcommand{\HM}{\mathrm{HM}}
\newcommand{\AS}{\mathrm{AS}}

\newtheorem{thm}{Theorem}[section]
\newtheorem{lem}[thm]{Lemma}
\newtheorem{pro}[thm]{Proposition}

\theoremstyle{remark}
\newtheorem{rem}[thm]{Remark}

\numberwithin{equation}{section}

\def\arg {\mathop{\rm arg}\nolimits}

\numberwithin{equation}{section}

\begin{document}

\title{On integrals of the tronqu\'{e}e solutions and the associated Hamiltonians for the Painlev\'{e} II equation}

\author{Dan Dai\footnotemark[1], ~Shuai-Xia Xu\footnotemark[2] ~and Lun Zhang\footnotemark[3]}

\renewcommand{\thefootnote}{\fnsymbol{footnote}}
\footnotetext[1]{Department of Mathematics, City University of Hong Kong, Tat Chee
Avenue, Kowloon, Hong Kong. E-mail: \texttt{dandai@cityu.edu.hk}}
\footnotetext[2]{Institut Franco-Chinois de l'Energie Nucl\'{e}aire, Sun Yat-sen University,
Guangzhou 510275, China. E-mail: \texttt{xushx3@mail.sysu.edu.cn}}
\footnotetext[3] {School of Mathematical Sciences and Shanghai Key Laboratory for Contemporary Applied Mathematics, Fudan University, Shanghai 200433, China. E-mail: \texttt{lunzhang@fudan.edu.cn }}
\date{}
\maketitle

\begin{abstract}
We consider a family of tronqu\'{e}e solutions of the Painelv\'{e} II equation
$$
q''(s)=2q(s)^3+sq(s)-(2\alpha+\frac12), \qquad \alpha > -\frac12,
$$
which is characterized by the Stokes multipliers
$$s_1=-e^{-2\alpha \pi i },\quad
s_2=\omega, \quad s_3=-e^{2 \alpha \pi i}
$$
with $\omega$ being a free parameter. These solutions include the well-known generalized Hastings-McLeod solution as a special case if $\omega=0$. We derive asymptotics of integrals of the tronqu\'{e}e solutions and the associated Hamiltonians over the real axis for $\alpha > -1/2$ and $\omega \geq 0$, with the constant terms evaluated explicitly. Our results agree with those already known in the literature if the parameters $\alpha$ and $\omega$ are chosen to be special values. Some applications of our results in random matrix theory are also discussed.
\end{abstract}

\tableofcontents

\noindent
\section{Introduction and statement of results}
In this paper, we are concerned with the Painlev\'{e} II equation defined by
\begin{equation}\label{eq:inhPII}
q''(s)=2q(s)^3+sq(s)-\nu, \qquad q(s)=q(s;\nu),
\end{equation}
where the parameter $\nu$ is a constant. This equation, together with other five nonlinear ordinary differential equations, was first introduced
by Painlev\'{e} and his colleagues around 1900 during the classification of second-order ordinary differential equations of the form
$$q''(s)=F(s,q,q')$$
with $F$ rational in $q$, $q'$, and analytic in $s$ except for isolated singularities in $\mathbb{C}$, whose solution possess the so-called Painlev\'{e} property, i.e., the only movable singularities of their solutions are poles; cf. \cite{InceBook} and \cite[\S 32.2]{DLMF}. Here `movable' means the locations of the singularities (which in general can be poles, essential singularities or branch points) of the solutions depend on the constants of integration associated with the initial or boundary conditions of the differential equations. For the Painlev\'{e} II equation \eqref{eq:inhPII}, it is known that all of its solutions are meromorphic in $s$ whose poles are simple with
residues $\pm 1$; cf. \cite[Section 2]{GLS}.

Although found more than one century ago, the Painlev\'{e} II equation \eqref{eq:inhPII} remains a research topic of great current interest due to its extensive connections with many branches of mathematics. For instance, just to name a few, similarity reductions of the modified Korteweg-de Vries equation satisfy \eqref{eq:inhPII}; see \cite{FN80}. The cumulative distribution function of the celebrated Tracy-Widom distribution \cite{TW}, which describes the limiting distribution of the largest eigenvalue for certain random matrix ensembles, can be represented as an integral involving the so-called Hastings-McLeod solution \cite{HM} of \eqref{eq:inhPII} with $\nu=0$; see also the survey article \cite{FW15} for the Painlev\'{e} II equation in random matrix theory and related fields. Hence, the solutions of Painlev\'{e} equations, also known as the Painlev\'{e} transcendents, are  widely recognized as the nonlinear counterparts of the classical special functions.

Like classical special functions, the Painlev\'e transcendents have many nice properties. One of the important features is that all of them can be written as a Hamiltonian system. In many cases, rather than the Painlev\'e transcendents, it is their Hamiltonians arise naturally in applications.  For example, it might be simpler to express the formula of the Tracy-Widom distribution function in terms of the Hamiltonian of the Painlev\'e II equation; see formula \eqref{eq:TW-PII-Ham} below. It is also worthwhile to mention that the Hamiltonians are closely related to the isomonodromic tau functions; see a recent paper \cite{Its:Prok} by Its and Prokhorov on the Hamiltonian properties of the tau functions and the references therein. Therefore, the study on properties of the associated Hamiltonians has attracted great attention as well.

\subsection{Smooth solutions of the Painlev\'{e} II equation on the real line and their related integrals}

A striking feature of the Painlev\'{e} II equation is that the total integral of many of its solutions or the related functions can be evaluated explicitly, possibly after a suitable regularization. This feature is not only of theoretical interest but also applicable to many practical problems arising from mathematical physics. To review the relevant results, recall that every Painlev\'{e} transcendent admits a Riemann-Hilbert (RH) characterization through the Stokes multipliers. In the case of the Painlev\'{e} II equation, the Stokes multiplier is a triple of parameters $(s_1,s_2,s_3)\in \mathbb{C}^3$ satisfying
\begin{equation}\label{eq:StokesMult}
s_1-s_2+s_3+s_1s_2s_3=-2 \sin (\nu \pi).
\end{equation}
Moreover, the map
\begin{equation}
\{\textrm{$(s_1,s_2,s_3)\in \mathbb{C}^3$ satisfying \eqref{eq:StokesMult}} \} \longrightarrow \{\textrm{solutions of \eqref{eq:inhPII}}\}
\end{equation}
is a bijection; cf. \cite{FN80,FIKNBook,JTU81}.

For the homogeneous Painlev\'{e} II equation, i.e., $\nu=0$ in \eqref{eq:inhPII}, by assuming that the solutions tend to zero exponentially fast as $s \to +\infty$, it is readily seen that the equation should be approximated by the classical Airy equation $y''(s)=sy(s)$. A result due to Hastings and McLeod \cite{HM} asserts that, for any $k$, there exists a unique solution to the homogeneous Painlev\'{e} II equation which behaves like $k\Ai(s)$ for large positive $s$, where $\Ai$ is the standard Airy function (cf. \cite[Chapter 9]{DLMF}). Depending on different choices of $k$, there are several classes of well-known solutions of the Painelv\'{e} II equation.
\begin{itemize}
  \item If $k=\pm 1$, the solutions, denoted by $q_{\HM}^{\pm}(s;0)$, are known as the Hastings-McLeod solutions \cite{HM}. They are determined by the choice of Stokes multipliers
\begin{equation}
s_1=\mp i, \qquad s_2=0, \qquad s_3=\pm i,
\end{equation}
with the asymptotic behavior
\begin{equation*}
q_{\HM}^{\pm}(s;0)=\left\{
                   \begin{array}{ll}
                     \pm \sqrt{-\frac{s}{2}}+O(|s|^{-5/2}), & \hbox{as $s \to -\infty$,} \\
                     \pm \Ai(s)+O\left(\frac{e^{-(4/3)s^{3/2}}}{s^{1/4}}\right), & \hbox{as $s \to +\infty$;}
                   \end{array}
                 \right.
\end{equation*}
see \cite{DZ95} for the detailed derivation of the above asymptotic formulas. The Hastings-McLeod solutions $q_{\HM}^{\pm}(s;0)$ are real and pole-free on the real axis. By subtracting off the non-integrable part near $s=-\infty$, the following total integrals of $q_{\HM}^{\pm}(s;0)$ are proved in \cite{BRD08,BBDI} via two different methods:
\begin{equation}\label{eq:intHHM}
\int_c^{+\infty}q_{\HM}^{\pm}(y;0)dy+\int_{-\infty}^c\left(q_{\HM}^{\pm}(y;0)\mp\sqrt{\frac{|y|}{2}}\right)dy
=\mp\frac{\sqrt{2}}{3}c|c|^{1/2}\pm\frac{\ln 2}{2}, \qquad c\in\mathbb{R}.
\end{equation}

Note that one of the explicit expressions for the Tracy-Widom distribution arising from the Gaussian Unitary Ensemble (GUE) \cite{TW,TW96} in random matrix theory is given in terms of  $q_{\HM}^{+}$ as follows:
\begin{equation} \label{eq:TW-PII}
  F_{\textrm{TW}}(x) = \exp \left( - \frac{1}{2} \int_x^{+\infty} (y-x) (q_{\HM}^{+}(y;0) )^2dy \right).
\end{equation}
The above formula can be rewritten as
\begin{equation} \label{eq:TW-PII-Ham}
  F_{\textrm{TW}}(x) = \exp \left( - \int_x^{+\infty} H(y) dy \right),
\end{equation}
where
\begin{equation} \label{eq:PIIHM-Hami}
H(x)=-2^{1/3}H_{\textrm{PII}}(-2^{1/3}x)
\end{equation}
with $H_{\textrm{PII}}$ being the Hamiltonian for the Painlev\'e II equation; see Forrester and  Witte \cite{FW01}, Tracy and Widom  \cite{TW} and \eqref{def:Hamiltonian-CPII}, \eqref{def:Hamiltonian-PII} below. These integrals are also used to describe the asymptotics of the distribution of the largest eigenvalue of a Gaussian Orthogonal Ensemble (GOE) and Gaussian Symplectic Ensemble (GSE) in the edge scaling limit; see \cite{TW96}.

As $x\to-\infty$, it is shown by Tracy and Widom in \cite{TW} that
\begin{equation} \label{eq:TW-large gap asy}
  \ln F_{\textrm{TW}}(x) =- \int_x^{+\infty} H(y)dy =-\frac{|x|^3}{12}-\frac{1}{8}\ln|x|+c_0+O(|x|^{-3/2}),
\end{equation}
where the constant term $c_0$ is conjectured to be
\begin{equation}\label{def:c0}
c_0=\frac{1}{24}\ln 2+\zeta'(-1)
\end{equation}
with $\zeta'(z)$ being the derivative of the Riemann zeta function. This conjecture was later proved by Deift, Its and Krasovsky in \cite{DIK2008}. From \eqref{eq:TW-large gap asy}, it is readily seen that the total integral of the Hamiltonian associated with the Hastings-Mcleod solution $q_{\HM}^{+}$, after certain regularization as $s \to-\infty$, can be evaluated explicitly by the Riemann zeta function. It is worthwhile to mention that the evaluation of the constant term in the asymptotics of the integral for the Hamiltonians associated with Painlev\'e transcendents is in general a challenging problem with important applications in mathematical physics. Several works have been contributed to the studies of this aspect \cite{BR18,dkv,E06,E10,T}; see also \cite{ bip,ilt2013,ilp2018,ilt2015,ip2016} for a novel approach developed recently to resolve this issue.

\item If $-1<k<1$, the one-parameter family of solutions, denoted by $q_{\AS}(s;0,k)$ are known as the real Ablowitz-Segur solutions \cite{AS1,AS2}. They are also pole-free on the real line \cite{AS3}, and determined by the Stokes multipliers
\begin{equation}
-1<is_1=k<1,\qquad s_3=\overline{s_1}=-s_1, \qquad s_2=0,
\end{equation}
with the asymptotic behavior
\begin{equation*}
q_{\AS}(s;0,k)= \begin{cases}
\frac{\sqrt{-2\chi}}{(-s)^{1/4}}\cos\left(\frac{2}{3}(-s)^{3/2}+\chi\ln(8(-s)^{3/2})+\phi\right) \\ \hspace{4.5cm} +O\left(\frac{\ln |s|}{|s|^{5/4}}\right), & \textrm{as } s\to -\infty, \\
k\Ai(s)+O\left(\frac{e^{-(4/3)s^{3/2}}}{s^{1/4}}\right), &  \textrm{as } s\to +\infty.
\end{cases}
\end{equation*}
Here,
\begin{equation*}
\chi:=\frac{1}{2\pi}\ln(1-|k|^2), \quad \phi:=-\frac{\pi}{4}-\arg\Gamma(i\chi)-\arg(-ki),
\end{equation*}
where $\Gamma(z)$ is the Gamma function. The following elegant total integral formula for $q_{\AS}$ is established in
\cite[Theorem 2.1]{BBDI}:
\begin{equation}
\int_{-\infty}^{+\infty}q_{\AS}(y;0,k)dy=\frac{1}{2}\ln\left(\frac{1+k}{1-k}\right).
\end{equation}
A total integral formula involving $q_{\AS}^2$ has recently been evaluated by the Barnes' $G$-function in \cite{BR18, bip}, which is related to the deformed Tracy-Widom distribution and the asymptotics of the Hankel determinant and orthogonal polynomials for a Gaussian weight with a discontinuity at the edge in a critical regime \cite{BCI}; see also Section \ref{application} below for more details about  this aspect.

\item There also exist purely imaginary Ablowitz-Segur solutions, which correspond to purely imaginary $k$. For the total integral formulas of these solutions as well as the generic purely imaginary solutions, we refer to \cite{BBDI} for more information.
\end{itemize}

For the inhomogeneous Painlev\'{e} II equation, i.e., $\nu \neq 0$ in \eqref{eq:inhPII}, the asymptotics of a real or purely imaginary Ablowitz-Segur solution is obtained in \cite{Dai:Hu}. Later, the total integral of these solutions is established \cite[Theorem 1.1]{Kokocki18} with an application to singularity formation in the vortex patch dynamics. We also refer to \cite{Miller} for the studies of total integral involving the increasing tritronque\'{e} solutions.

The aim of the present work is devoted  to the studies of integrals of Painlev\'{e} II transcendents and the corresponding Hamiltonians by considering an important one-parameter family of solutions of \eqref{eq:inhPII}, namely, the tronqu\'{e}e solutions, as we will introduce in what follows.

\subsection{Tronqu\'{e}e solutions of the Painlev\'{e} II equation}
By setting
\begin{equation}\label{eq:paranu}
\nu=2\alpha+\frac12
\end{equation}
in \eqref{eq:inhPII}, the special family of Painlev\'{e} II transcendents we are interested in, denoted by $q(s;2\alpha+\frac12,\omega)$, corresponds to the Stokes multipliers
\begin{equation} \label{eq:Stokes-Tron}
s_1=-e^{-2\alpha \pi i }, \qquad s_2=\omega, \qquad s_3=-e^{2 \alpha \pi i},
\end{equation}
where $\omega$ is a free parameter. As shown in \cite{Kapaev} (see also \cite[Chapter 11]{FIKNBook}), these solutions belong to the classical tronqu\'{e}e solutions due to Boutroux \cite{Boutroux}, which means they are pole-free near infinity in one or more sectors of opening angle $2\pi /3$ of the complex plane. Their applications in random matrix theory and asymptotics of orthogonal polynomials can be found in \cite{CKV08,wxz}.

If $\omega=0$, these solutions are exactly the generalized Hastings-McLeod solutions characterized by the following asymptotic behaviors:
\begin{equation} \label{eq:HMqa-asy}
q(s;2\alpha+\frac12,0)=\left\{
                         \begin{array}{ll}
                           \sqrt{-\frac{s}{2}}-\frac{\alpha+\frac14}{s}+O(|s|^{-5/2}), & \hbox{as $s\to -\infty$,} \\
                           \frac{2\alpha+\frac12}{s}+O(s^{-4}), & \hbox{as $s \to +\infty$.}
                         \end{array}
                       \right.
\end{equation}
Like the homogeneous case, the generalized Hastings-McLeod solutions are real and pole-free on the real axis for $\alpha>-1/2$; see \cite{CKV08}. In addition, we have
\begin{equation}\label{eq:qs00}
q(s;0,0)=q_{\HM}^{+}(s;0),
\end{equation}
where $q_{\HM}^{+}$ is one of the Hastings-McLeod solutions to the homogeneous Painlev\'{e} II equation mentioned before.

For general $\omega \neq 0$, $q(s;2\alpha+\frac12,\omega)$ is smooth near $- \infty$ and has the same asymptotic expansion as $q(s;2\alpha+\frac12,0)$. However, $q(s;2\alpha+\frac12,\omega)$ has infinitely many poles on the real axis accumulating at $+\infty$; see formulas \eqref{def:w} and \eqref{eq:asy-w-negative infty-omega} below. The dependence of the parameter $\omega$ can actually be seen from the coefficients of the oscillatory terms in the asymptotics along the rays $\arg s = 2\pi/3$ or $\arg s = 4\pi/3$, which reflects the quasi-linear Stokes phenomenon for the Painlev\'{e} II transcendent; cf. \cite[Remark 11.5]{FIKNBook}. It is also worthwhile to point out that if $\alpha=0$ and $\omega=1$, one has
\begin{equation}\label{q-special solution}
q(s;\frac12,1)=-2^{-1/3}\frac{\Ai'(-2^{-1/3}s)}{\Ai(-2^{-1/3}s)};
\end{equation}
cf. \cite[Equation (11.4.15)]{FIKNBook}.

With motivations arising from random matrix theory (cf. Section \ref{application} below), we will consider the case that the parameter
$\alpha>-1/2$ and the Stokes multiplier $\omega \geq 0$. We next state our results regarding the regularized integrals of $q(s;2\alpha+\frac12,\omega)$ and the associated Hamiltonians.

\subsection{Statement of results}
In order to obtain a convergent integral, similar to \eqref{eq:intHHM}, we subtract the leading term in the asymptotic expansion \eqref{eq:HMqa-asy} and define an integral for $q(s;2\alpha+\frac{1}{2},0)$ as follows:
\begin{align}\label{def:I11-new}
I_1(s;\alpha,0)&:=\int^{c}_{-\infty} \left( q(\tau;2\alpha+\frac{1}{2},0)-\sqrt{-\frac{\tau}{2}}+\frac{\alpha+\frac{1}{4}}{\tau} \right) d\tau
\nonumber \\
 & \quad + \int_{c} ^{s}q(\tau;2\alpha+\frac{1}{2},0)d\tau +\frac {\sqrt{2}}{3} c |c|^{\frac 12}-(\alpha+\frac{1}{4}) \ln |c|, \quad c<0.
\end{align}
If $\omega \neq 0$, $q(s;2\alpha+\frac12,\omega)$ may have real poles in principle. Since all of these poles are simple with residues $\pm 1$, we then make a modification in the second integral of the above formula and define
\begin{align}\label{def:I12-new}
I_1(s;\alpha,\omega)&:=\int^{c}_{-\infty} \left( q(\tau;2\alpha+\frac{1}{2},\omega)-\sqrt{-\frac{\tau}{2}}+\frac{\alpha+\frac{1}{4}}{\tau} \right) d\tau
\nonumber \\
 & \quad + \textrm{P.V.} \int_{c} ^{s}q(\tau;2\alpha+\frac{1}{2},\omega)d\tau +\frac {\sqrt{2}}{3} c |c|^{\frac 12}-(\alpha+\frac{1}{4}) \ln |c|, \quad c<0,
\end{align}
where $c$ is smaller than $s$ and  such that $q(\tau;2\alpha+\frac{1}{2},\omega)$ is pole-free on $(-\infty, c)$, and P.V. stands for the Cauchy principal value. Our first result is the following theorem.
\begin{thm}\label{thm:totalintPII}
For $\alpha>-1/2$ and $\omega \geq 0$, let $q(s;2\alpha+\frac{1}{2},\omega)$ be the tronqu\'{e}e solutions of the Painlev\'{e} II equation \eqref{eq:inhPII} with $\nu=2\alpha+\frac12$ and characterized by the Stokes multipliers given in \eqref{eq:Stokes-Tron}. Then, the regularized integrals $I_1(s;\alpha,\omega)$ defined in \eqref{def:I11-new} and \eqref{def:I12-new} are convergent and independent of the parameter $c$.
Furthermore, we have, as $s \to +\infty$,
\begin{equation}\label{eq:total integral PII1}
I_1(s;\alpha,0)= (2\alpha+\frac{1}{2})\ln s +\frac{1}{2}\ln (2\pi)-\ln \Gamma(1+2\alpha)- (\alpha+\frac{1}{4})\ln 2+O\left(\frac{1}{s^{3/2}}\right),
\end{equation}
and
\begin{align} \label{eq:total integral PII2}
  I_1(s;\alpha,e^{-2\beta \pi i})  & =
 (\frac{\alpha}{2} - \frac{1}{4})\ln s + \ln\left|\cos\left(\frac{\vartheta(s)}{2}+\arg\Gamma(1+\alpha-\beta)-\frac{\pi}{4}\right)\right|+ \frac{\beta}{2} \pi i
 \nonumber \\
 & \quad +( \frac{5}{3} \alpha+1) \ln 2
 +\ln\left(\frac{|\Gamma(1+\alpha-\beta)|}{\Gamma(1+2\alpha)}\right)
 +O\left(\frac{1}{s^{3/2}}\right),
\end{align}
where $\beta i \in \mathbb{R}$ and
\begin{equation} \label{def:theta}
 \vartheta(s)=\vartheta(s;\alpha,\beta) := \frac{2 \sqrt{2}}{3} s^{3/2} - 3 \beta i \ln s - \alpha \pi - 5  \beta i\ln 2 .
\end{equation}

\end{thm}

\begin{rem}
When $\alpha=-1/4$ and $\omega=0$, we have that $q(s;0,0)=q_{\HM}^{+}(s;0)$. From \eqref{def:I11-new}
and \eqref{eq:total integral PII1}, it follows that
\begin{equation*}
\int_{-\infty} ^{c} \left( q_{\HM}^+(\tau;0)-\sqrt{-\frac{\tau}{2}} \right) d\tau + \int_{c} ^{+\infty}q_{\HM}^+(\tau;0)d\tau
= \frac{\ln 2}{2} - \frac {\sqrt{2}}{3} c |c|^{\frac 12}, \qquad c<0,
\end{equation*}
which is indeed  \eqref{eq:intHHM} established in \cite{BRD08,BBDI}. Here, since the term $(\alpha + \frac{1}{4}) /\tau$ in the first integral of \eqref{def:I11-new} vanishes for $\alpha = -1/4$, the constant $c$ in the above formula can actually be taken to be an arbitrary real number.
\end{rem}

We next move to the Hamiltonian associated with the tronqu\'ee solutions. In view of the Hamiltonian representation of the Tracy-Widom distribution shown in \eqref{eq:TW-PII-Ham} and \eqref{eq:PIIHM-Hami}, we consider the following rescaled Hamiltonian
dynamical system:
\begin{equation}\label{eq:Hamiltonian system}
\left\{\begin{array}{l}
         \displaystyle -\frac{dw}{ds}=\frac{\partial H}{\partial u} =w^2-2u-s, \vspace{1.5mm} \\
         \displaystyle \frac{du}{ds}=\frac{\partial H}{\partial w} = 2uw+2\alpha,
       \end{array}
\right.
\end{equation}
where
\begin{equation}\label{def:Hamiltonian-CPII}
H(s; 2\alpha, \omega)=H(u,w,s ; 2\alpha)=-u^2+(w^2-s)u+2\alpha w
\end{equation}
is the Hamiltonian. Here, the functions $u(s)=u(s;2\alpha, \omega)$ and $w(s) = w(s;2\alpha+\frac12,\omega)$ are related to $q(s)=q(s;2\alpha+\frac12,\omega)$ through the formulas
\begin{align}
w(s) &=  -2^{1/3}q(-2^{1/3}s), \label{def:w} \\
u(s) &=  2^{-1/3}U(-2^{1/3}s),\quad U(s)=q'(s)+q(s)^2+\frac{s}{2} \label{def:u}.
\end{align}
Note that $u(s)$ satisfies the Painlev\'{e} XXXIV equation \cite{InceBook}
\begin{equation}\label{eq:P34}
u^{\prime\prime}(s)=4u(s)^2+2su(s)+\frac{(u'(s))^2-\gamma^2}{2u(s)}
\end{equation}
with the parameter $\gamma=2\alpha$; cf. \cite[Equations (5.0.2) and (5.0.55)]{FIKNBook}. The Hamiltonian \eqref{def:Hamiltonian-CPII} has appeared in recent work on GUE with boundary spectrum singularity \cite{wxz} and is also simply related to the standard Hamiltonian for the Painlev\'{e} II equation \eqref{eq:inhPII} (cf. \cite[Equation 32.6.9]{DLMF})
\begin{equation}\label{def:Hamiltonian-PII}
H_{\textrm{PII}}(p,q,s; \nu= 2\alpha+\frac12)=\frac{1}{2}p^2-pq^2-\frac{s}{2}p+2\alpha q
\end{equation}
via the transformation \eqref{eq:PIIHM-Hami}.
The Hamiltonian appearing in \eqref{eq:TW-PII-Ham} corresponds to $\alpha=\omega=0$. In addition, it is the rescaled Hamiltonian $H(s)$ that satisfies the following Jimbo-Miwa-Okamoto $\sigma$-form \cite{JTU81} of the Painlev\'{e} II equation
\begin{equation} \label{PII-sigma-form}
(\sigma''(s))^2+4\sigma'(s)\left((\sigma'(s))^2-s\sigma'(s)+\sigma(s)\right)-\gamma^2=0
\end{equation}
with $\gamma=2\alpha$; cf. \cite[Chapter 8]{Forrester}. In particular, if $\alpha=0$ and $\omega=1$, we obtain from the special solution \eqref{q-special solution} that
\begin{equation}\label{w-special-solution}
w(s;\frac{1}{2},1)=\frac{\Ai'(s)}{\Ai(s)}.
\end{equation}
This, together with the first equation in \eqref{eq:Hamiltonian system}, implies that
\begin{equation}\label{u-special-solution}
u(s;0,1)=0.
\end{equation}  Thus, it follows from  \eqref{def:Hamiltonian-CPII}  that the Hamiltonian  associated with this special solution is also trivial, i.e.,
\begin{equation}\label{H-special-solution}
H(s; 0,1)=0.
\end{equation}

We now define a regularized integral of $H(s;2\alpha,\omega)$ by
\begin{multline}\label{def:I-2}
I_2(s;\alpha,\omega)
:=\int_{c} ^{+\infty} \left( H (\tau;2\alpha,\omega)+2\alpha\sqrt{\tau}
+\frac{\alpha^2}{\tau} \right) d\tau
\\+\int_{s} ^{c}H(\tau;2\alpha,\omega)d\tau
+\frac 43 \alpha c^{\frac 32}+\alpha^2 \ln c, \quad c>0.
\end{multline}
Comparing with \eqref{def:I11-new} and \eqref{def:I12-new}, the above integral is integrated toward $+\infty$ instead of from $-\infty$, due to the minus sign in the change of variable in  \eqref{eq:PIIHM-Hami}. Moreover, unlike the tronqu\'ee solutions, the Hamiltonian $H(s)$ is pole-free on the real line; see \cite[Theorem 1]{wxz}. Therefore, no Cauchy principal value is needed in the definition of $I_2$ for all $\omega \geq 0$.

Our second result is the following theorem.
\begin{thm}\label{thm:Hamil}
For $\alpha>-1/2$ and $\omega \geq 0$, the regularized integral $I_2(s;\alpha,\omega)$ of the Hamiltonian $H(s;2\alpha,\omega)$ given in \eqref{def:I-2} is convergent and independent of $c$. Furthermore, $I_2(s;\alpha,\omega)$ admits the following asymptotics as $s\to -\infty$. For $\omega=0$, we have
\begin{equation}\label{eq:total integral H}
I_2(s;\alpha,0)=-\frac {s^3}{12} -(2\alpha^2-\frac {1}{8}) \ln|s|+  \alpha
-\frac{\ln 2}{24}-\zeta'(-1)
+\ln \left(\frac{G(1+2\alpha)}{(2\pi)^{\alpha}}\right)  +O\left(\frac{1}{|s|^{3/2}}\right),
\end{equation}
where $\zeta'(z)$ is the derivative of the Riemann zeta-function and $G(z)$ is the Barnes' $G$-function. For $\omega=e^{-2\beta \pi i}$ with $\beta i \in \mathbb{R}$, we have
\begin{align}\label{thm: H-integral}
  I_2(s;\alpha,e^{-2\beta \pi i}) & = \frac{4}{3}i\beta |s|^{3/2}+\frac{3\beta^2-\alpha^2}{2} \ln|s|+3(\beta^2-\alpha^2)\ln2-\alpha\beta\pi i
\nonumber \\
& \quad -\ln\left(\frac{G(1+\alpha+\beta)G(1+\alpha-\beta)}{G(1+2\alpha)}\right)+O\left(\frac{\ln |s|}{|s|^{3/2}}\right).
\end{align}
\end{thm}

\begin{rem}
If $\alpha=0$ and $i\beta\in[0,+\infty)$, our result \eqref{thm: H-integral} reads
\begin{align}\label{I-3-alpha-0-omega}
  I_2(s;0,e^{-2\beta \pi i}) = \frac{4}{3}i\beta |s|^{3/2}+\frac{3\beta^2}{2} \ln|s|+3\beta^2\ln2-\ln G(1+\beta)G(1-\beta)+O\left(\frac{\ln |s|}{|s|^{3/2}}\right),
\end{align}
which agrees with \cite[Equation (1.34)]{bip}.  A detailed discussion of this comparison will be given in Section \ref{application} below.

\end{rem}


\subsection{Applications}\label{application}
In this section, we present two applications of our results to random matrix theory.

\subsubsection*{Largest eigenvalue distribution of a perturbed GUE with a discontinuity at the edge}
Let us consider a perturbed GUE defined by the following joint probability density function of the $n$ eigenvalues $\lambda_1<\lambda_2<\ldots<\lambda_n$:
\begin{equation}\label{def:jpdf}
\frac{1}{Z_{n,\mu,\omega}}\prod_{i=1}^n w(\lambda_i;\mu, \omega) \prod_{1 \leq i < j \leq n}(\lambda_i-\lambda_j)^2,
\end{equation}
where $Z_{n,\mu,\omega}$ is the normalization constant, and
\begin{equation}\label{def:weight}
w(x;\mu, \omega)=e^{-x^2}\left\{
                           \begin{array}{ll}
                             1, &  \hbox{$x \leq \mu$,} \\
                             \omega, & \hbox{$x>\mu$,}
                           \end{array}
                         \right. \qquad \omega \geq 0,
\end{equation}
is the Gaussian weight with a jump singularity at $\mu$. If $\omega=1$, the ensemble \eqref{def:jpdf} reduces to the well-known GUE. For large $n$, the eigenvalues of GUE are distributed over the interval $[-\sqrt{2n},\sqrt{2n}]$ and the local statistics obeys the principle of universality; cf. \cite{metha}. If $\omega  \neq 1$, a natural question is then to establish the local behavior of the eigenvalues around the discontinuous point $\mu$ of the weight function $w$.

Clearly, the location of $\mu$ plays an important role in the relevant studies. By assuming that
\begin{equation}
\mu=\mu_n=\sqrt{2n}+\frac{s}{\sqrt{2}n^{1/6}},
\end{equation}
i.e., the point of discontinuity is close to $\sqrt{2n}$ (which is the soft edge of the unperturbed GUE), it is shown in \cite{BCI} that the limiting distribution of the largest eigenvalue $\lambda_n$ is given by
\begin{align}
F(s)&:=\lim_{n\to\infty}\textrm{Prob}\left(\lambda_n \leq \mu_n \right)
\nonumber \\
& =\exp\left(-\int_s^{+\infty}(\tau-s)(q_{\AS}(\tau;0,\sqrt{1-\omega}))^2d\tau\right)
\nonumber \\
&=\det\left(1-(1-\omega)K_{\Ai}|_{[s,+\infty)}\right),
\end{align}
 where $q_{\AS}(s;0,\sqrt{1-\omega})$, $0<\omega<1$, is the real Ablowitz-Segur solution of the homogeneous Painlev\'{e} II equation corresponding to the Stokes multiplier $s_1=-i\sqrt{1-\omega}$ as introduced at the beginning, and the third equality involving the Fredholm determinant can be found in \cite{TW}. Here, $K_{\Ai}$ is the trace-class integral operator with kernel
\begin{equation}\label{def:Airykernel}
K_{\Ai}(x,y)=\frac{\Ai'(x)\Ai(y)-\Ai(y)\Ai'(x)}{x-y}
\end{equation}
acting on $[s,+\infty)$. The distribution function $F$ has also appeared in the thinned process of GUE \cite{Boh1,Boh2}. Moreover, by investigating the asymptotics of Hankel determinant associated with the discontinuous function \eqref{def:weight}, the following asymptotics of $F(s)$ as $s \to -\infty$ (also known as the large gap expansion) is conjectured in \cite[Conjecture 3]{BCI}:
\begin{equation}\label{eq:conject}
\lim_{s\to-\infty}\left( \ln F(s)+\frac{4}{3}i\beta|s|^{3/2}+\frac{3}{2}\beta^2\ln|s|\right)
=\ln(G(1+\beta)G(1-\beta))-3\beta^2\ln 2,
\end{equation}
for $\omega=e^{-2\beta \pi i}$, $i\beta\in \mathbb{R}$.
This formula was later proved in \cite{BR18,bip}. Our asymptotic formula \eqref{thm: H-integral} actually provides another proof. Indeed, by \cite[Corollary 1]{wxz}, it follows that
\begin{equation}
\ln F(s)=\int_s^{+\infty} \biggl[H(\tau; 0,1 ) - H(\tau; 0, \omega) \biggr]d\tau ,
\end{equation}
where $H$ is the Hamiltonian defined in \eqref{def:Hamiltonian-CPII}. As $H(\tau; 0, 1) \equiv 0$ (see \eqref{H-special-solution}), by setting $\alpha=0$ in \eqref{thm: H-integral}, we recover \eqref{eq:conject}.

\subsubsection*{Phase transition from the soft edge to the hard edge -- a view from large gap asymptotics}

For the second application of our result, we start with a concrete unitary random matrix ensemble
\begin{equation}\label{def:shiftedGUE}
\frac{1}{Z_{n,N}}(\det M)^{2\alpha}\exp(-N \textrm{Tr} V_c(M))dM, \quad \alpha>-\frac{1}{2},
\end{equation}
defined on $n\times n$ positive definite Hermitian matrices $M$, where
\begin{equation}\label{def:weight-Gaussian-shifted and scaled}
V_c(x)=\frac{1}{2c}(x-2)^2,  \qquad x\geqslant0, \qquad c>0.
\end{equation}
 It is well-known that the limiting mean eigenvalue density $\psi_{V_c}$ of $M$ as $n, N\to \infty$ and $n/N\to 1$ is described by the unique minimizer of the energy functional
$$
\iint \log \frac{1}{|x-y|}\psi(x)\psi(y)dxdy+\int V_c(x) \psi(x)dx,
$$
where $\psi$ is taken over all the probability density function supported on $[0,+\infty)$; cf. \cite{D}. Thus, $\psi_{V_c}$ is also known as the equilibrium measure.

It comes out that the above equilibrium problem can be solved explicitly, and as the parameter $c$ varies, one encounters a phase transition from the soft edge to the hard edge. More precisely, if $c \leq 1$, we have the semi-circle law
\begin{equation}
\psi_{V_c}(x)=\frac{1}{2\pi c}\sqrt{4c-(x-2)^2},\quad x\in[2-2\sqrt{c},2+2\sqrt{c}].
\end{equation}
Note that the left ending point of the support $2(1-\sqrt{c})$ is strictly positive for $c<1$, which is the so-called soft edge. As $c \to 1^{-}$ , it approaches the origin, which is the so-called hard edge. If $c>1$, we have
\begin{equation}\label{def:density-MP-law}
\psi_{V_c}(x)=\frac{(x+a)}{2\pi c}\sqrt{\frac{b-x}{x}},\quad x\in[0,b],
\end{equation}
where
$$a=-\frac{4}{3}+\frac{2}{3}\sqrt{1+3c},\qquad b=\frac{4}{3}+\frac{4}{3}\sqrt{1+3c}.$$
In this case, the support of $\psi_{V_c}$ always contains the hard edge and $\psi_{V_c}$ has a square-root singularity at $0$. Similar phase transition occurs if one replaces the potential $V_c$ in \eqref{def:shiftedGUE} with a general real analytic potential $V$ and requires the associated equilibrium measure vanishes like a square root at the origin.

A question now is how to describe this phase transition. One way to tackle this problem is to explore the scaling limit of the correlation kernel characterizing the eigenvalue distribution. It is well-known that the scaling limit of the correlation kernel leads to the Airy kernel \eqref{def:Airykernel} and the Bessel kernel near the soft edge and the hard edge, respectively, which indicates the principle of universality. The question then asks about the transitional case from the Airy kernel to the Bessel kernel. When $n, N\to \infty$ and $n/N \to 1$ in such a way that $n/N-1=O(n^{-2/3})$, it is shown by Claeys and Kuijlaars in \cite{CK} that the scaling limit of the correlation kernel near the origin is described by the Painlev\'{e} XXXIV kernel, which is also closely related to the generalized Hastings-Mcleod solutions of \eqref{eq:inhPII} and the limiting kernels found in \cite{CKV08}.

Our result actually provides a heuristic interpretation of the phase transition from a different perspective, namely, the large gap asymptotics. The idea is based on the fact that the limiting distributions of the (re-scaled) smallest eigenvalue of the ensemble \eqref{def:shiftedGUE} are given by the Fredholm determinants of the limiting kernels for the correlation kernel near the origin. The phase transition should also be observed from the asymptotics of the relevant Fredholm determinants, i.e., the large gap asymptotics. As the parameter varies, we indeed find such kind of transition from the non-trivial constant term therein. More precisely, we start with a fact that the limiting  kernel in \cite{CK} can be viewed as a special case of the general Painlev\'{e} XXXIV kernel $K_{\alpha,\omega}^{P34}(x,y;t)$ with the parameters $\alpha>-\frac 12$, $\omega=0$, where $t$ comes from the limit $c_1n^{2/3}(1-\frac nN ) \to t$ for some constant $c_1$; see \cite{ikj2008,wxz}.  For $\alpha>-\frac 12$, $\omega\geqslant 0$ and  $t\in \mathbb{R}$, let $K^{P34}_{\alpha,\omega}|_{[s,+\infty)}$ be the trace-class operator acting on $L^2(s,+\infty)$ with the Painlev\'{e} XXXIV kernel $K_{\alpha,\omega}^{P34}(x,y;t)$. It is shown by the first two authors of the present work that, as $s\to -\infty$,
 \begin{align}
   \ln \det\biggl( I-K^{P34}_{\alpha,\omega}|_{[s,+\infty)} \biggr)=&-\frac 1{12}|s+t|^3 +\frac 23 \alpha|s|^{\frac 32}-2\alpha|s|^{\frac{1}{2}}t-(\alpha^2+\frac {1}{8}) \ln|s+t| \label{P34: lag gap asy} \nonumber \\
 &+\frac 43 \alpha \, \mathrm{sgn}(t) |t|^{\frac 32}  +\alpha^2\ln|t| \nonumber \\
 &+\int_t^{+\infty}(\tau-t)\left(u(\tau)-\frac {\alpha}{|\tau|^{\frac{1}{2}}}+\frac {\alpha^2}{\tau^2}\right )d\tau +c_0 +o(1),
 \end{align}
 where $u(s)=u(s;2\alpha,\omega)$ is the Painlev\'{e} XXXIV transcendent \eqref{eq:P34} and $c_0$ is the constant given \eqref{def:c0}; see \cite[Theorem 4]{xd}.

As mentioned before, since $K_{\alpha,0}^{P34}(x,y;t)$ describes the transition from the soft edge to the hard edge as $t$ varies from $+\infty$ to $- \infty$, it is expected to observe this interesting transition from the asymptotics \eqref{P34: lag gap asy} as well. Let us focus on the integral in \eqref{P34: lag gap asy}. With the relation $\frac{dH}{ds} = -u(s)$ (see \eqref{eq:H-u} below) and the large $s$ asymptotics of $u(s)$ (see Propositions \ref{Pro: asy large s} and \ref{Pro: asy phi negative infty} below), we obtain from an integration by parts that
\begin{equation}
  \int_t^{+\infty}(\tau-t)\left(u(\tau)-\frac {\alpha}{\tau^{\frac{1}{2}}}+\frac {\alpha^2}{\tau^2}\right )d\tau = - \int_{t} ^{+\infty} \left( H (\tau;2\alpha,\omega)+2\alpha\sqrt{\tau}
+\frac{\alpha^2}{\tau} \right) d\tau
\end{equation}
for $t>0$. It is easily seen from the large positive $s$ asymptotics of $H(\tau;2\alpha,\omega)$ given in \eqref{eq:H-asy} below that, when $t \to +\infty$, the above integral tends to 0. As a consequence, the remaining constant term in \eqref{P34: lag gap asy} is $c_0$, which is exactly the constant term in the asymptotics of Fredholm determinant associated with the Airy kernel; see \cite{DIK2008,TW}. When $t \to -\infty$, after a modification as that in the definition of $I_2$ in \eqref{def:I-2} and using its asymptotics in \eqref{eq:total integral H}, the non-trivial constant term in \eqref{P34: lag gap asy}  becomes
$$
\ln \left(\frac{G(1+2\alpha)}{(2\pi)^{\alpha}}\right).
$$
It comes out that $\ln \left(\frac{G(1+2\alpha)}{(2\pi)^{\alpha}}\right)$ is just the constant term in the asymptotics of a Bessel-kernel determinant; see \cite{dkv,E10,TW94-2}. It would be interesting to give a rigorous proof of the above heuristic observation, which describes a phase transition from the soft edge to the hard edge from the viewpoint of large gap asymptotics. We will leave this problem to a further investigation.

\subsection{Organization of the rest of this paper}

The rest of this paper is devoted to the proofs of our results via the RH approach \cite{D}. In Section \ref{sec:RH problem}, we recall an RH problem $\Psi$ for the Painlev\'e XXXIV equation. Its connections with the Hamiltonian $H$ and the integral of the Painlev\'{e} II transcedent $w$ are established. To relate the integral of the Hamiltonian to $\Psi$, we derive several differential identities for the Hamiltonian with respect to the parameter $\alpha$ and the Stoke's multiplier $\omega$, similar to those obtained in \cite{bip, ilp2018,Its:Prok}. Based on the asymptotic analysis of the RH problem for $\Psi$ carried out in \cite{ikj2009,wxz}, we obtain the asymptotics of $w,H$ and other relevant functions as $s \to \pm \infty$ in Section \ref{sec:asymptotics}. We present the proofs of Theorems \ref{thm:totalintPII} and \ref{thm:Hamil} in Section \ref{Sec:Main-proof}. For the convenience of the reader, we include various local parametrices used in the asymptotic analysis of the RH problem for $\Psi$ in the Appendix.

\section{RH problem for the Painlev\'e XXXIV equation and differential identities}\label{sec:RH problem}
We first recall the  following RH  problem $\Psi(\zeta;s)$ for  the Painlev\'e XXXIV equation; see \cite{fik, ikj2008,ikj2009}.

\subsection*{RH problem for $\Psi$}
\begin{description}
  \item(a)  $\Psi(\zeta)=\Psi(\zeta;s)$ is a $2\times 2$  matrix-valued function, which is analytic for $\zeta$ in
  $\mathbb{C}\backslash \{\cup_{j=1}^4\Sigma_j\cup\{0\}\}$, where
\begin{align*}
\Sigma_1=(0,+\infty), ~~ \Sigma_2=e^{\frac{2 \pi i}{3}}(0,+\infty), ~~\Sigma_3=(-\infty,0),~~\Sigma_4=e^{-\frac{2 \pi i}{3}}(0,+\infty),
\end{align*}
with the orientations as shown in Figure \ref{fig:jumpsPsi}.

\begin{figure}[t]
\begin{center}
   \setlength{\unitlength}{1truemm}
   \begin{picture}(100,70)(-5,2)
       \put(40,40){\line(-2,-3){16}}
       \put(40,40){\line(-2,3){16}}
       \put(40,40){\line(-1,0){30}}
       \put(40,40){\line(1,0){30}}

       \put(30,55){\thicklines\vector(2,-3){1}}
       \put(30,40){\thicklines\vector(1,0){1}}
       \put(50,40){\thicklines\vector(1,0){1}}
       \put(30,25){\thicklines\vector(2,3){1}}

       \put(39,36.3){$0$}
       \put(20,11){$\Sigma_4$}
       \put(20,69){$\Sigma_2$}
       \put(3,40){$\Sigma_3$}

       \put(72,40){$\Sigma_1$}
       \put(25,44){$\Omega_2$}
       \put(25,34){$\Omega_3$}
       \put(55,44){$\Omega_1$}
       \put(55,33){$\Omega_4$}

       \put(40,40){\thicklines\circle*{1}}

   \end{picture}
   \caption{The jump contours $\Sigma_j$ and the regions $\Omega_j$, $j=1,2,3,4$ for the RH problem for $\Psi$.}
   \label{fig:jumpsPsi}
\end{center}
\end{figure}
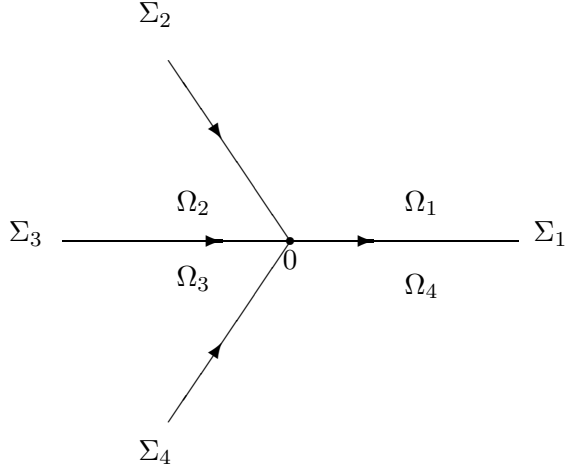

  \item(b)  $\Psi(\zeta)$  satisfies the jump condition
  \begin{gather}\label{eq:Psi-jump}
  \Psi_+ (\zeta)=\Psi_- (\zeta)
  \left\{ \begin{array}{ll}
           \begin{pmatrix}
                                 1 & \omega \\
                                 0 & 1
                                 \end{pmatrix}, &\quad \zeta \in {\Sigma}_1,
\\
           \begin{pmatrix}
                                 1 & 0 \\
                                 e^{2\alpha\pi i} & 1
                                 \end{pmatrix}, &\quad \zeta \in {\Sigma}_2,
\\
           \begin{pmatrix}
                                 0 & 1 \\
                                 -1 & 0
                                 \end{pmatrix},& \quad
                                           \zeta \in {\Sigma}_3,
\\
            \begin{pmatrix}
                                 1 & 0 \\
                                 e^{-2\alpha\pi i} & 1
                                \end{pmatrix},  & \quad \zeta \in \Sigma_4.
          \end{array}
    \right .
  \end{gather}
\item(c)   As $\zeta\rightarrow \infty$, there exists a function $a(s)$ such that
\begin{equation} \label{eq:Psi-infinity}
 \Psi(\zeta;s) = \begin{pmatrix}
                      1 & 0\\
                      ia(s) & 1
                   \end{pmatrix}
       \left (I+\frac {\Psi_1(s)}{\zeta}
    +O\left(  \zeta^{-2}\right) \right)
 \zeta^{-\frac{1}{4}\sigma_3}  \frac{I+i\sigma_1}{\sqrt{2}} e^{-(\frac{2}{3}\zeta^{3/2}+s\zeta^{1/2}) \sigma_3},
\end{equation}
  where the functions $\zeta ^{1/4}, \zeta ^{1/2}, \zeta ^{3/2}$ take the principle branches and
\begin{equation}\label{def:Paulimatrice}
\sigma_1=\begin{pmatrix}
 0 & 1 \\
 1 & 0
 \end{pmatrix}, \qquad \sigma_3=\begin{pmatrix}
 1 & 0 \\
 0 & -1
 \end{pmatrix}
\end{equation}
are the Pauli matrices. Furthermore, we have
$$(\Psi_1)_{12}(s)=-ia(s),$$
where $(M)_{ij}$ stands for the $(i,j)$-th entry of a matrix $M$.
\item(d) As $\zeta\rightarrow 0$, there exist constant matrices $E_j$, $j=1,2,3,4$, such that
\begin{equation}\label{eq:Psi0-1}
 \Psi(\zeta;s)=\Psi_0(\zeta;s)\zeta^{\alpha \sigma_3} E_j,~~\zeta\in \Omega_j, \quad  \mbox{if}~2\alpha\notin   \mathbb{N}_0:=\{0\}\cup \mathbb{N}, \alpha>-\frac {1}{2},
 \end{equation}
 \begin{equation}\label{eq:Psi0-2}
\Psi(\zeta;s)= \Psi_0(\zeta;s)\zeta^{\alpha \sigma_3}\left (I+\frac{\kappa}{2\pi i}\ln \zeta ~\sigma_+\right ) E_j,~~\zeta\in \Omega_j,  \quad \mbox{if}~2\alpha\in   \mathbb{N}_0,
\end{equation}
where  the regions $\Omega_j$, $j=1,2,3,4$ are shown in Figure \ref{fig:jumpsPsi},
\begin{equation}\label{def:kappa}
\kappa=e^{2\pi i\alpha}-\omega, \qquad \sigma_+=
                   \begin{pmatrix}
                     0 & 1
                     \\
                     0 & 0
                   \end{pmatrix},
\end{equation}
and the above multi-valued functions take the principle branches with the cuts along the negative axis.
Here, $\Psi_0(\zeta;s)$ is analytic at $\zeta=0$ and has the following expansion
\begin{equation} \label{Psi0-origin-exp}
  \Psi_0(\zeta;s) = \Phi_0(s) + \Phi_1(s) \zeta + \cdots.
\end{equation}
The piecewise  constant matrices $E_i$ are defined by
\begin{equation} \label{E2-def}
  E_2=\left\{\begin{array}{ll}
              \begin{pmatrix}
        \frac{ e^{2\alpha\pi i}\omega-1}{2ie^{\alpha\pi i}\sin(2\alpha\pi)} & \frac{e^{2\alpha\pi i}-\omega}{2ie^{\alpha\pi i}\sin(2\alpha\pi)} \\
         -e^{\alpha\pi i} & e^{-\alpha\pi i}
       \end{pmatrix}, & \quad 2\alpha\notin   \mathbb{N}_0,\alpha>-\frac {1}{2},\\
               \begin{pmatrix}
                     1 & 0 \\
                     -1 & 1
                  \end{pmatrix}, & \quad  \alpha\in \mathbb{N}_0, \\
               \begin{pmatrix}
               0 & -1 \\
                     1 & 1
                  \end{pmatrix},  & \quad \alpha-\frac {1}{2}\in \mathbb{N}_0,
             \end{array}
\right.
\end{equation}
and the other $E_j$, $j=1,3,4$, are given through the relations
\begin{equation}
  E_1=E_4J_1,\quad E_1=E_2J_2, \quad E_3=E_4J_4,
\end{equation}
where $J_i$ denotes the jump matrix for $\Psi$ on $\Sigma_i$ for $i=1,2,3,4$; see \cite[Proposition 2.2]{ikj2008}.
\end{description}

We have the following lemma concerning the solvability of the RH problem for $\Psi$ and its relation to the Hamiltonian $H$ defined in \eqref{def:Hamiltonian-CPII}.

\begin{lem}\label{lem:solpsi}
For $\alpha>-\frac 12$, $\omega\geq 0$
and $s\in\mathbb{ R}$, there exists
a unique solution to the above RH problem for $\Psi$. The Hamiltonian $H$ defined in \eqref{def:Hamiltonian-CPII} is related to
$\Psi$ through the relation
\begin{equation}\label{eq:H-Psi}
H(s;2\alpha,\omega)=\frac{s^2}{4}+a(s)=\frac{s^2}{4}+i(\Psi_1)_{12}(s),
\end{equation}
where $a(s)$ is given in the asymptotics of $\Psi$ near infinity in \eqref{eq:Psi-infinity}. Furthermore, $H(s;2\alpha,\omega)$ is  real and  pole-free for $s \in \mathbb{R}$.
\end{lem}
\begin{proof}
The existence of unique solution to the above RH problem is proved in \cite[Proposition 1]{wxz} and the relation \eqref{eq:H-Psi} is given in \cite[Equation (44)]{wxz}, which also implies that $H(s)$ is pole-free for $s\in\mathbb{R}$.

To show $H(s)$ is real for $s\in\mathbb{R}$, we set
\begin{equation} \label{eq:Psi-hat}
\widehat{\Psi}(\zeta;s)= \begin{pmatrix}
                      1 & 0\\
                      -ia(s) & 1
\end{pmatrix}\Psi(\zeta;s).
\end{equation}
Then, $\widehat{\Psi}(\zeta;s)$  is the unique solution to the above RH problem but with the boundary condition at infinity \eqref{eq:Psi-infinity} replaced by
$$\widehat{ \Psi}(\zeta;s) =
       \left (I
    +O\left(1/  \zeta\right) \right) \zeta^{-\frac{1}{4}\sigma_3}  \frac{I+i\sigma_1}{\sqrt{2}} e^{-(\frac{2}{3}\zeta^{3/2}+s\zeta^{1/2}) \sigma_3},\qquad \zeta \to \infty.$$
Since $\sigma_3\overline{\widehat{ \Psi}(\bar{\zeta};s)}\sigma_3$ satisfies the same RH problem as
$\widehat{\Psi}(\zeta;s)$, where $\overline{f(\zeta)}$ denotes the complex conjugation of $f$, it then follows from the unique solvability of the RH problem that
\begin{equation} \label{eq:Psi-symmetry}
\sigma_3\overline{\widehat{ \Psi}(\bar{\zeta};s)}\sigma_3=\widehat{\Psi}(\zeta;s).
\end{equation}
A combination of \eqref{eq:H-Psi}, \eqref{eq:Psi-hat} and  \eqref{eq:Psi-symmetry} yields that $H(s)$ is real for $s\in \mathbb{R}$.

This completes the proof of Lemma \ref{lem:solpsi}.
\end{proof}

Besides the Hamiltonian $H$, one could also relate the rescaled Painlev\'e II transcendent $w(s) = w(s;2\alpha+\frac12,\omega)$ to the RH problem for $\Psi$, as shown in the next lemma.

\begin{lem}\label{lem: integrals of w}
Let $\Phi_0(s)$ be the first term in the expansion of the analytic part of $\Psi(\zeta;s)$ near $\zeta = 0$ (cf. \eqref{Psi0-origin-exp}), then  
\begin{equation} \label{eq:w-anti-derivative}
	 w(s;2\alpha+ \frac{1}{2}, \omega)= \frac{d}{ds}\ln \left|\left(\Phi_0\right)_{11}(s)\right| .
\end{equation}
\end{lem}
\begin{proof}
Recall the following Lax pair associated with the function $\Psi(\zeta;s)$ (see \cite[Lemma 3.2]{ikj2008}):
\begin{align}
  \frac{\partial}{ \partial \zeta} \Psi(\zeta;s) & = \begin{pmatrix}
    \frac{u'}{2\zeta} & i - i\frac{u}{\zeta} \\ -i \zeta - i(u+s) - i \frac{(u')^2 - (2\alpha)^2}{4u\zeta} & -\frac{u'}{2\zeta}
  \end{pmatrix} \Psi(\zeta;s), \\
  \frac{\partial}{ \partial s} \Psi(\zeta;s) & =
  \begin{pmatrix}
    0 & i
    \\
    -i \zeta - 2i(u + \frac{s}{2}) & 0
  \end{pmatrix} \Psi(\zeta;s) , \label{lax-pair-Psi-s}
\end{align}
where $u$ is given in \eqref{def:u}. From \eqref{lax-pair-Psi-s}  and the asymptotic behavior of $\Psi(\zeta;s)$  near the origin given in \eqref{eq:Psi0-1} and \eqref{eq:Psi0-2},
we have
\begin{equation}\label{eq:deff-Psi-0}
\frac{d}{ds}\Phi_0(s)=\begin{pmatrix}
                              0 & i \\
                              -2i(u+\frac{s}{2})& 0
                            \end{pmatrix}\Phi_0(s).
\end{equation}
This, together with \eqref{eq:Hamiltonian system}, implies that
\begin{equation}\label{eq:phi-0}
\frac {d^2}{ds^2}\left(\Phi_0\right)_{11}(s)=2(u+\frac{s}{2})\left(\Phi_0\right)_{11}(s)=(w^2+w')\left(\Phi_0\right)_{11}(s).
\end{equation}
From \eqref{Psi0-origin-exp}, \eqref{eq:Psi-hat} and  \eqref{eq:Psi-symmetry}, we see that  $\left(\Phi_0\right)_{11}(s)$ is real for $s\in \mathbb{R}$.
By setting
$$y(s)=\ln\left|\left(\Phi_0\right)_{11}(s)\right|, $$
it follows from \eqref{eq:phi-0} that
\begin{equation}
(y')^2 + y'' = w^2+w'.
\end{equation}
From the large $s$ asymptotics of $w$ and $\left(\Phi_0\right)_{11}$ derived in Section \ref{sec:asymptotics} below, we conclude from the above
equation that $y'=w$, i.e., the identity \eqref{eq:w-anti-derivative}.

This completes the proof of Lemma \ref{lem: integrals of w}.
\end{proof}


We finally present some useful differential identities related to the Hamiltonian $H$.
\begin{lem}\label{lem: total diff}
Let $H(s)=H(s;2\alpha,\omega)$ be the Hamiltonian  defined in \eqref{def:Hamiltonian-CPII}. Then, we have
\begin{equation}\label{eq:H-u}
H_s=-u,
\end{equation}
\begin{equation}\label{eq:total differential}
(uw+2sH)_{s}=3(uw_{s}+H)+3(H-2\alpha w),
\end{equation}
and
\begin{equation}\label{eq:differential alpha}
(uw_{s}+H)_{\alpha}=(uw_{\alpha})_{s}+2w, \quad (uw_{s}+H)_{\omega}=(uw_{\omega})_{s}
\end{equation}
where the subscripts $s$, $\alpha$  and $\omega$ denote the derivative with respect to the variable $s$, the parameter $\alpha$ and the Stokes data $\omega$, respectively.
\end{lem}
\begin{proof}
We obtain from the Hamiltonian system \eqref{eq:Hamiltonian system} that
\begin{equation}\label{eq:H-uw'}
H_s=w_s\frac{\partial H}{\partial w}+u_s\frac{\partial H}{\partial u}-u=-u,
\end{equation}
which is \eqref{eq:H-u}. Also note that
\begin{equation}
wu_s-2uw_s=w\frac{\partial H}{\partial w}+2u\frac{\partial H}{\partial u}=2w(w+\alpha)-2u(w^2-2u-s)
=4H+2su-6\alpha w,
\end{equation}
which gives us \eqref{eq:total differential} by combining with \eqref{eq:H-u}.
Similarly, a straightforward calculation yields
\begin{multline}
(uw_{s}+H)_{\alpha}=u_{\alpha}w_{s}+uw_{s\alpha}+w_{\alpha}\frac{\partial H}{\partial w}+u_{\alpha}\frac{\partial H}{\partial u}+2w
\\
=u_{\alpha}w_{s}+uw_{s\alpha}+w_{\alpha}u_s-u_{\alpha}w_s+2w
=(uw_{\alpha})_{s}+2w,
\end{multline}
and
\begin{multline}
(uw_{s}+H)_{\omega}=u_{\omega}w_{s}+uw_{s\omega}+w_{\omega}\frac{\partial H}{\partial w}+u_{\omega}\frac{\partial H}{\partial u}
\\
=u_{\omega}w_{s}+uw_{s\omega}+w_{\omega}u_s-u_{\omega}w_s =(uw_{\omega})_{s}.
\end{multline}

This completes the proof of the Lemma \ref{lem: total diff}.
\end{proof}

As we will see later, the above differential identities play an important role in the proof of our main results. Roughly speaking, these identities allow us to represent the regularized integrals of the Painlev\'e II transcendents and the associated Hamiltonians or their derivatives with respect to various parameters explicitly into some known functions. It is crucial that the asymptotics of these functions can be found, as shown in the next section.


\section{Large $s$ asymptotics of $w$, $H$ and relevant functions}\label{sec:asymptotics}

In this section, we will derive the asymptotics of the rescaled Painlev\'e II transcendent $w(s) = w(s;2\alpha+\frac12,\omega)$,  the associated Hamiltonian $H(s) = H(s; 2\alpha, \omega)$ and other relevant functions as $s \to \pm \infty$. The derivation is based on the asymptotic analysis of the RH problem for $\Psi$ carried out in \cite{ikj2009,wxz}.

\subsection{Asymptotic analysis of the RH problem for $\Psi$  for large positive $s$}\label{sec:RH analysis large postivs s}

The asymptotic analysis of the RH problem for $\Psi$  as $s \to +\infty$ was performed by Its, Kuijlaars and \"Ostensson in \cite{ikj2009} for the special case $\alpha>-\frac 12$ and $\omega=1$, which can be extended to the case $\alpha>-\frac 12$ and arbitrary  $\omega\geq 0$.   For our purpose, we give a sketch of the analysis here.

Following the idea of nonlinear steepest descent analysis of the RH problem \cite{D}, the analysis consists of a series of explicit and invertible transformations:
\begin{equation}
\Psi \to A \to B \to C \to D.
\end{equation}
More precisely, we have
\begin{itemize}
  \item The first transformation $\Psi \to A$ is defined by
\begin{equation}\label{A}
A(\zeta)=s^{\sigma_3/4}\begin{pmatrix}
1 & 0\\
 -ia(s)& 1
  \end{pmatrix} \Psi(s\zeta;s),
\end{equation}
 which is a scaling of the variable.
  \item The second transformation $A \to B$ creates an RH problem for $B$ whose jump contour $\Sigma_B$ consists of four rays starting from $z=-1$ instead of from $z=0$; see \cite[Equation (2.4)]{ikj2009}. In particular, if $\zeta$ is small and belongs to the region $\Omega_2$ (see Figure \ref{fig:jumpsPsi}), the transformation is given by
\begin{equation}\label{def:BtoA}
B(\zeta)=A(\zeta)
J_2=A(\zeta)
\begin{pmatrix}
1 & 0
\\
e^{2\alpha \pi i} & 1
\end{pmatrix}.
\end{equation}
  \item The third transformation $B \to C$ is defined by
\begin{equation}\label{C}
C(\zeta)=
\begin{pmatrix}
1 & 0 \\
-\frac{is^{3/2}}{4} & 1
\end{pmatrix} B(\zeta)e^{s^{3/2}g(\zeta)\sigma_3},
\end{equation}
where
\begin{equation}\label{g-function}
  g(\zeta) = \frac{2}{3} (\zeta+1)^{3/2}, \qquad - \pi < \arg(\zeta+1) < \pi.
\end{equation}
As a consequence, the asymptotic behavior of $B$ at infinity is normalized.
\item We now construct various approximations of $C$ in different regions of the complex plane for large positive $s$. Since the jump matrices of $C$ tend to the identity matrices exponentially fast except the ones on $(-\infty, -1)$ and $(-1,0)$ (on which the jumps are constant matrices) as $s\to +\infty$, it is expected that $C$ is well-approximated by a global parametrix $P^{(\infty)}$ that solves the following RH problem.
\subsubsection*{RH problem for $P^{(\infty)}$}
\begin{description}
  \item(a)  $P^{(\infty)}(\zeta)$ is a $2\times 2$  matrix-valued function, which is analytic for $\zeta$ in
  $\mathbb{C}\backslash (-\infty,0]$.
  \item(b)  $P^{(\infty)}(\zeta)$  satisfies the jump condition
  \begin{gather}
  P_+^{(\infty)}(\zeta)=P_-^{(\infty)}(\zeta)
  \left\{ \begin{array}{ll}
           \begin{pmatrix}
                                 0 & 1 \\
                                 -1 & 0
                                 \end{pmatrix}, & \zeta \in (-\infty,-1),
\\
                                  \begin{pmatrix}
                                 e^{2\alpha \pi i} & 0
                                 \\
                                 0 & e^{-2\alpha \pi i}
                                 \end{pmatrix}, & \zeta \in (-1,0).
          \end{array}
    \right .
  \end{gather}
\item(c)   As $\zeta\rightarrow \infty$, we have
\begin{equation}
 P^{(\infty)}(\zeta) =
       (I+O(1/\zeta))  \zeta^{-\frac{1}{4}\sigma_3}\frac{I+i\sigma_1}{\sqrt{2}}.
\end{equation}
\end{description}
By \cite[Equation (2.16)]{ikj2009}, the solution to the above RH problem is given by
\begin{equation}\label{P-infty}
P^{(\infty)}(\zeta)=\begin{pmatrix}
    1 & 0 \\ 2 \alpha i & 1
  \end{pmatrix} (\zeta + 1)^{-\sigma_3/4} \frac{1}{\sqrt{2}} \begin{pmatrix}
    1 & i \\ i & 1
  \end{pmatrix} d(\zeta)^{\sigma_3},
\end{equation}
where
\begin{equation}\label{def:d}
d(\zeta):=\left( \frac{(\zeta+1)^{1/2} + 1}{(\zeta+1)^{1/2} - 1} \right)^{-\alpha}
\end{equation}
with the branches of the arguments fixed by the inequalities
$$-\pi <\arg(\zeta+1)<\pi, \qquad -\pi<\arg \left( \frac{(\zeta+1)^{1/2} + 1}{(\zeta+1)^{1/2} - 1} \right)<\pi.$$
The global parametrix $P^{(\infty)}$ does not provide a good approximation near $\zeta=-1$ and $\zeta=0$, where special treatments are needed.
Near $\zeta=-1$, $C$ is approximated by the local parametrix
\begin{equation}\label{eq:Airypara}
P^{(-1)}(\zeta)=E(\zeta)\Phi^{(\Ai)}(s(1+\zeta))e^{s^{3/2}g(\zeta)\sigma_3},
\end{equation}
where $E$ is analytic near $\zeta=-1$, $\Phi^{(\Ai)}$ is the well-known Airy parametrix (see Section \ref{sec:Airy} below) and $g$ is defined in \eqref{g-function}. Near $\zeta=0$, we need to construct a local parametrix $P^{(0)}$ solving the following RH problem.
\subsubsection*{RH problem for $P^{(0)}$}
\begin{description}
  \item(a)  $P^{(0)}(\zeta)$ is a $2\times 2$  matrix-valued function, which is analytic for $\zeta$ in
  $U(0,\delta)\backslash (-\delta,\delta)$, for some  $\delta\in(0,1)$, where $U(a,\delta)$ stands for an open disc centered at $a$ with radius $\delta$.
  \item(b)  $P^{(0)}(\zeta)$  satisfies the jump condition
  \begin{gather}
  P_+^{(0)}(\zeta)=P_-^{(0)}(\zeta)
  \left\{ \begin{array}{ll}
           \begin{pmatrix}
                                 1& \omega e^{-2s^{3/2}g(\zeta)} \\
                                 0 & 1
                                 \end{pmatrix}, & \quad \zeta \in (0,\delta),
\\
                                  \begin{pmatrix}
                                 e^{2\alpha \pi i} & e^{-2s^{3/2}g(\zeta)}
                                 \\
                                 0 & e^{-2\alpha \pi i}
                                 \end{pmatrix}, & \quad  \zeta \in (-\delta,0).
          \end{array}
    \right .
  \end{gather}
\item(c)   As $s\rightarrow +\infty$, we have the matching condition
\begin{equation}
 P^{(0)}(\zeta) =
       (I+O(e^{-\frac{2}{3}s^{3/2}})) P^{(\infty)}(\zeta) ,
\end{equation}
for $\zeta\in \partial U(0,\delta)$, where $P^{(\infty)}$ is the global parametrix given in \eqref{P-infty}.
\end{description}

The solution to the above RH problem is explicitly given by
\begin{equation}\label{P-0}
P^{(0)}(\zeta)=  \begin{pmatrix}
    1 & 0 \\ 2 \alpha i & 1
  \end{pmatrix} (\zeta+1)^{-\sigma_3/4} \frac{1}{\sqrt{2}} \begin{pmatrix}
    1 & i \\ i & 1
  \end{pmatrix}
\begin{pmatrix}
    1 & r(\zeta)\\
    0&1
  \end{pmatrix}  d(\zeta)^{ \sigma_3},
\end{equation}
where $d(\zeta)$ is defined in \eqref{def:d} and $r(\zeta)$ is expressed in terms of the Cauchy integrals
\begin{equation}\label{def:r}
r(\zeta):=\frac{1}{2\pi i}\left(\int_{-1}^0\frac{e^{-2s^{3/2}g(\tau)}|d(\tau)|^2}{\tau-\zeta} d\tau+\omega \int_0^{\infty}\frac{e^{-2s^{3/2}g(\tau)}|d(\tau)|^2}{\tau-\zeta} d\tau\right).
\end{equation}
For later use, we need to know the asymptotic behavior of $r(\zeta)$ near the origin. By using \eqref{def:d} and \eqref{def:r}, we have, as $\zeta \to 0$ and $0<\arg \zeta<\pi$,
\begin{equation}\label{eq:r-0-case 1}
r(\zeta)=\left\{
           \begin{array}{ll}
             r_1(\zeta)+
                \frac{1-\omega e^{-2\alpha\pi i}}{2i\sin(2\alpha\pi)}e^{-2s^{3/2}g(\zeta)}d^2(\zeta), & \quad \hbox{$\alpha>-1/2$ and $2\alpha \notin \mathbb{N}_0$,}
\vspace{1.5mm}\\
             r_2(\zeta)+\frac{  e^{2\alpha\pi i}-\omega}{2\pi i} e^{-2s^{3/2}g(\zeta)}d^2(\zeta)\ln \zeta, & \quad \hbox{$2\alpha\in \mathbb{N}_0$,}
           \end{array}
         \right.
\end{equation}
where $r_1(\zeta)$ and $r_2(\zeta)$ are two analytic functions near the origin.
%
%

\item In the final transformation $C \to D$, we compare the global and local parametrices by considering
\begin{equation} \label{D}
 D(\zeta) = \left\{
              \begin{array}{ll}
                C(\zeta)\left(P^{(\infty)}(\zeta)\right)^{-1}, & \quad \hbox{$\zeta \in  \mathbb{C}\setminus (\overline{U(0,\delta)}\cup \overline{U(-1,\delta)})$,} \\
                C(\zeta)\left(P^{(0)}(\zeta)\right)^{-1}, & \quad \hbox{$\zeta \in U(0,\delta)$,} \\
                C(\zeta)\left(P^{(-1)}(\zeta)\right)^{-1}, & \quad \hbox{$\zeta \in U(-1,\delta)$.}
              \end{array}
            \right.
\end{equation}
Then, $D$ also satisfies an RH problem which is analytic in $\mathbb{C}\setminus \Sigma_D$, where $$\Sigma_D:=\partial U(-1,\delta)\cup \partial U(0,\delta) \cup (-1+\delta,-\delta)\cup (\delta,+\infty) \cup e^{\frac{2\pi i}{3}}(-1+\delta,+\infty)
\cup e^{-\frac{2\pi i}{3}}(-1+\delta,+\infty).$$
Since $D_+(\zeta)=D_-(\zeta)J_D(\zeta)$ for $\zeta \in \Sigma_D$ with
$$J_D=I+O(s^{-3/2}),\qquad s\to +\infty,$$ uniformly on $\Sigma_D$, and $$D(\zeta)=I+O(\zeta^{-1}),\qquad  \zeta \to \infty,$$
it follows from a standard argument as in \cite{D} that
\begin{equation}\label{D-est}
  D(\zeta) = I+O(s^{-3/2}), \qquad s\to +\infty,
\end{equation}
uniformly for $\zeta \in \mathbb{C} \setminus \Sigma_D$; see \cite[Equation (2.28)]{ikj2009}.
\end{itemize}

We are now ready to derive the asymptotics of $w,H$ and relevant functions for large positive $s$, which will be the topics of the next section.

 \subsection{Asymptotics of $w,H$ and relevant functions for large positive $s$}

\begin{pro}\label{Pro: asy large s}
For $\alpha>-\frac 12$ and  $\omega \geq 0$, 
we have the following asymptotics
as $s\to +\infty$:
\begin{align}
u(s;2\alpha,\omega)& =\frac{\alpha}{\sqrt{s}} \left(1-\frac{\alpha}{s^{3/2}}+O(s^{-3}) \right) \nonumber \\
& \quad +(e^{2\pi i\alpha } - \omega)\frac{\Gamma(2\alpha+1)}{2^{2+6\alpha} \pi}s^{-(3\alpha+ \frac{1}{2})}e^{- \frac{4}{3} s^{3/2}}(1+O(s^{-1/4})), \label{eq:u-asy} \\
w(s;2\alpha+\frac {1}{2},\omega) &=-\sqrt{s}-\frac{\alpha+\frac{1}{4}}{s}+O(s^{-5/2}), \label{eq:w-asy} \\
H(s;2\alpha,\omega) &=-2\alpha\sqrt{s}-\frac{\alpha^2}{s}+O(s^{-5/2}), \label{eq:H-asy} \\
\ln\left|\left(\Phi_0\right)_{11}(s)\right|& = -\frac{2}{3}s^{3/2}-(\alpha+\frac{1}{4})\ln s
-(2 \alpha + \frac{1}{2})\ln 2+O(s^{-3/2}). \label{eq:asy-phi-0}
\end{align}
\end{pro}
\begin{proof}
The asymptotics for the Painlev\'e XXXIV transcendent $u(s;2\alpha,\omega)$ and the Hamiltonian $H(s;2\alpha,\omega)$ in \eqref{eq:u-asy} and \eqref{eq:H-asy} are proved in \cite[Theorem 1.2]{ikj2009}, \cite[Theorem 2]{wxz} and \cite[Equation (17)]{wxz}, respectively.
The asymptotics of the tronqu\'ee  Painlev\'e II transcendent $w(s;2\alpha,\omega)$ in \eqref{eq:w-asy} is known in \cite[Equation (11.5.56)]{FIKNBook}. The expansion in \eqref{eq:w-asy} except the error term can also be obtained by substituting  \eqref{eq:u-asy} into the second equation in \eqref{eq:Hamiltonian system}.

To show the asymptotic behavior of $\ln|\left(\Phi_0\right)_{11}(s)|$, we first observe from \eqref{eq:Psi0-1}--\eqref{Psi0-origin-exp} that
\begin{align} \label{Phi0-limit-relation}
  \Phi_0(s) = \begin{cases}
    \displaystyle \lim_{\zeta \to 0 , \atop \zeta \in \Omega_j} \Psi(\zeta; s) E_j^{-1} \zeta^{-\alpha \sigma_3}, & \quad 2\alpha\notin   \mathbb{N}_0, \alpha>-\frac {1}{2}, \vspace{1.5mm} \\
    \displaystyle \lim_{\zeta \to 0, \atop \zeta \in \Omega_j} \Psi(\zeta; s) E_j^{-1} \left (I-\frac{\kappa}{2\pi i}\ln \zeta~ \sigma_+\right ) \zeta^{-\alpha \sigma_3}, & \quad 2\alpha\in   \mathbb{N}_0.
  \end{cases}
\end{align}
The above formulas also hold if $\zeta$ is replaced by $s\zeta$ or $-s\zeta$ on the right hand side. Moreover, if $\zeta \in \Omega_2$ and small,
it follows from the transformations \eqref{A}--\eqref{C} and \eqref{D} that
\begin{align}
  \Psi(s\zeta; s) & = \begin{pmatrix}
    1 & 0 \\ ia(s) & 1
  \end{pmatrix}
s^{-\sigma_3/4} \begin{pmatrix}
    1 & 0 \\ \frac{is^{3/2}}{4} & 1
  \end{pmatrix} D (\zeta) \begin{pmatrix}
    1 & 0 \\ 2 \alpha i & 1
  \end{pmatrix} (\zeta+1)^{-\sigma_3/4}  \nonumber \\
  & \quad \times  \frac{1}{\sqrt{2}} \begin{pmatrix}
    1 & i \\ i & 1
  \end{pmatrix} \begin{pmatrix}
    1 &r(\zeta)\\
    0&1
  \end{pmatrix} d(\zeta)^{ \sigma_3} e^{-s^{3/2} g(\zeta) \sigma_3} J_2^{-1}, \label{Psi-sz-in-II}
\end{align}
where $g$ is defined in \eqref{g-function} and $D(z)$ satisfies the estimate \eqref{D-est}.

We next split our discussions into two cases, based on different values of $\alpha$. If $2\alpha\notin   \mathbb{N}_0$ and $\alpha>-\frac {1}{2}$, we obtain from \eqref{Phi0-limit-relation} and \eqref{Psi-sz-in-II} that
\begin{align}
  \Phi_0(s)  =  \begin{pmatrix}
    1 & 0 \\ ia(s)+\frac{i s^2}{4}& 1
  \end{pmatrix}
s^{-\sigma_3/4}  D(0)
  \begin{pmatrix}
    1 & 0 \\ 2 \alpha i & 1
  \end{pmatrix} \frac{1}{\sqrt{2}} \begin{pmatrix}
    1 & i \\ i & 1
  \end{pmatrix}
  \begin{pmatrix}
    1 & r_1(0) \\ 0 & 1
  \end{pmatrix}M_1(s)s^{-\alpha\sigma_3}, \label{Phi0-expression}
\end{align}
where
\begin{equation}\label{def:K}
M_1(s):=\lim_{\zeta \to 0 , \atop \zeta \in \Omega_2}
\begin{pmatrix}
    1 & r(\zeta)-r_1(\zeta)\\
    0&1
  \end{pmatrix}d(\zeta)^{ \sigma_3} e^{-s^{3/2} g(\zeta) \sigma_3} (E_2J_2)^{-1}\zeta^{-\alpha\sigma_3},
\end{equation}
with $r(\zeta)$ and $r_1(\zeta)$ given in \eqref{eq:r-0-case 1}. To this end, we note from \eqref{eq:r-0-case 1} that
\begin{equation} \label{eq:r-d}
d(\zeta)^{ -\sigma_3} e^{s^{3/2} g(\zeta) \sigma_3}\begin{pmatrix}
    1 & r(\zeta)-r_1(\zeta)\\
    0&1
  \end{pmatrix}d(\zeta)^{ \sigma_3} e^{-s^{3/2} g(\zeta) \sigma_3}
  =\begin{pmatrix}
    1 & \frac{1-\omega e^{-2\alpha\pi i}}{2i\sin(2\alpha\pi)} \\
    0 &1  \end{pmatrix},\quad \zeta\in \Omega_2,
\end{equation}
and by \eqref{eq:Psi-jump} and \eqref{E2-def},
\begin{equation} \label{E2-J2-1}
E_2 J_2= e^{\alpha \pi i\sigma_3}\begin{pmatrix}
   1 & \frac{1-\omega e^{-2\alpha\pi i}}{2i\sin(2\alpha\pi)} \\
    0 & 1
  \end{pmatrix}.
  \end{equation}
 Substituting \eqref{eq:r-d} and \eqref{E2-J2-1} in \eqref{def:K} yields
\begin{equation}\label{eq:M-1-exp}
 M_1(s)= \lim_{\zeta \to 0 , \atop \zeta \in \Omega_2} d(\zeta)^{ \sigma_3} e^{-s^{3/2} g(\zeta) \sigma_3} e^{-\alpha \pi i\sigma_3}\zeta^{-\alpha\sigma_3}= 2^{-2\alpha \sigma_3}
    e^{-\alpha \pi i\sigma_3}  e^{-\frac{2}{3}s^{3/2} \sigma_3},
\end{equation}
since $d(\zeta)\sim (\zeta/4)^{\alpha}$ as $\zeta \to 0$. This, together with \eqref{Phi0-expression} and \eqref{D-est}, implies that, as $s\to +\infty$,
\begin{multline}
  \Phi_0(s)  =  \begin{pmatrix}
    1 & 0 \\ ia(s)+\frac{i s^2}{4}& 1
  \end{pmatrix}
s^{-\sigma_3/4}  (I+O(s^{-3/2}))
  \begin{pmatrix}
    1 & 0 \\ 2 \alpha i & 1
  \end{pmatrix}
\\
\times \frac{1}{\sqrt{2}} \begin{pmatrix}
    1 & i \\ i & 1
  \end{pmatrix}
  \begin{pmatrix}
    2^{-2\alpha}e^{-\pi i \alpha} e^{-\frac23 s^{2/3}} & \ast \\ 0 & \ast
  \end{pmatrix}s^{-\alpha\sigma_3},
\end{multline}
where $\ast$ stands for some unimportant term. We then obtain \eqref{eq:asy-phi-0} for  $\alpha>-1/2$ and $2\alpha\not\in \mathbb{N}_0$.

If $2\alpha\in  \mathbb{N}_0$, in view of \eqref{Phi0-limit-relation}, the expression in \eqref{Phi0-expression} is modified to be
\begin{align}
  \Phi_0(s) & =  \begin{pmatrix}
    1 & 0 \\ ia(s)+\frac{i s^2}{4}& 1
  \end{pmatrix}
s^{-\sigma_3/4}  D(0)
  \begin{pmatrix}
    1 & 0 \\ 2 \alpha i & 1
  \end{pmatrix} \frac{1}{\sqrt{2}} \begin{pmatrix}
    1 & i \\ i & 1
  \end{pmatrix}
  \begin{pmatrix}
    1 & r_2(0) \\ 0 & 1
  \end{pmatrix}M_2(s)s^{-\alpha\sigma_3}, \label{Phi0-expression-2}
\end{align}
where
\begin{equation}\label{def:M-2}
M_2(s):=\lim_{\zeta \to 0 , \atop \zeta \in \Omega_2}
\begin{pmatrix}
    1 & r(\zeta)-r_2(\zeta)\\
    0&1
  \end{pmatrix}d(\zeta)^{ \sigma_3} e^{-s^{3/2} g(\zeta) \sigma_3} (E_2J_2)^{-1}\left (I-\frac{\kappa}{2\pi i}\ln \zeta ~\sigma_+\right )\zeta^{-\alpha\sigma_3}
\end{equation}
with $r(\zeta)$ and $r_2(\zeta)$ defined in \eqref{eq:r-0-case 1} and $\kappa$ defined in \eqref{def:kappa}.
Since
\begin{equation} \label{eq:r-d-2}
d(\zeta)^{ -\sigma_3} e^{s^{3/2} g(\zeta) \sigma_3}\begin{pmatrix}
    1 & r(\zeta)-r_2(\zeta)\\
    0&1
  \end{pmatrix}d(\zeta)^{ \sigma_3} e^{-s^{3/2} g(\zeta) \sigma_3}
  =\begin{pmatrix}
    1 & \frac{  e^{2\alpha\pi i}-\omega}{2\pi i} \ln \zeta\\
    0 &1  \end{pmatrix}, \quad \zeta\in  \Omega_2,
\end{equation}
and
\begin{align} \label{E2-J2-2}
  E_2 J_2 = \begin{cases}
    \begin{pmatrix}
      1 & 0 \\ -1 & 1
    \end{pmatrix} \begin{pmatrix}
      1 & 0 \\ e^{2\alpha\pi i} & 1
    \end{pmatrix} = I, & \quad \alpha \in \mathbb{N}_0, \vspace{2mm} \\
    \begin{pmatrix}
      0 & -1 \\ 1 & 1
    \end{pmatrix} \begin{pmatrix}
      1 & 0 \\ e^{2\alpha\pi i} & 1
    \end{pmatrix} = \begin{pmatrix}
      1 & -1 \\ 0 & 1
    \end{pmatrix}, & \quad \alpha- \frac{1}{2}\in \mathbb{N}_0,
  \end{cases}
\end{align}
it follows that
\begin{equation}\label{eq:M-2-exp}
 M_2(s)= 2^{-2\alpha \sigma_3}  e^{-\frac{2}{3}s^{3/2} \sigma_3}, \qquad  2\alpha\in \mathbb{N}_0.
\end{equation}
Finally, substituting \eqref{eq:M-2-exp} in \eqref{Phi0-expression-2}, we again obtain \eqref{eq:asy-phi-0} for $2\alpha\in\mathbb{N}_0$ from \eqref{D-est}.

This completes the proof of Proposition \ref{Pro: asy large s}.
\end{proof}

\begin{rem} Actually, the asymptotic analysis of the RH problem for $\Psi$ as $s\to +\infty$ carried out  in Section \ref{sec:RH analysis large postivs s} holds for  $\alpha>-1/2$ and more general parameter $\omega\in \mathbb{C}\setminus (-\infty,0)$. Thus, Proposition \ref{Pro: asy large s} is also true for  $\alpha>-1/2$ and $\omega\in \mathbb{C}\setminus (-\infty,0)$.
\end{rem}

\subsection{Asymptotic analysis of the RH problem for $\Psi$  for large negative $s$}
The detailed nonlinear steepest descend analysis of the RH problem for $\Psi$ as $s\to -\infty$ was performed in
\cite{ikj2009} for $\alpha>-1/2$ and $\omega>0$, and by Wu, Xu and Zhao in  \cite{wxz} for the case $\alpha>-1/2$ and $\omega=0$. Again, we give a sketch of the analysis in this section for later use.

\subsubsection*{The case $\alpha>-1/2$ and $\omega=0$}
Following \cite{wxz}, the analysis consists of a series of explicit and invertible transformations:
\begin{equation}
\Psi \to \widetilde A \to \widetilde B \to D_{\alpha}.
\end{equation}
More precisely, we have
\begin{itemize}
  \item The first transformation $\Psi \to \widetilde A$ is defined by
  \begin{equation}\label{A-0}
\widetilde A(\zeta)=(-s)^{\sigma_3/4}\begin{pmatrix}
1 & 0\\
 -ia(s)& 1
  \end{pmatrix} \Psi(-s\zeta;s),
\end{equation}
 which is a scaling of the variable.
  \item The second transformation $\widetilde A \to \widetilde B$ is defined by
  \begin{equation}\label{B-0}
\widetilde B(\zeta)=\widetilde A(\zeta) e^{-|s|^{3/2}h(\zeta)\sigma_3},
\end{equation}
where
\begin{equation}\label{h-function}
  h(\zeta) = \zeta^{1/2}-\frac{2}{3}\zeta^{3/2}, \qquad -\pi< \arg \zeta < \pi.
\end{equation}
The aim of this transformation is to normalize the asymptotic behavior of $\widetilde A$ at infinity.
  \item As $s\to -\infty$, the jump matrices for $\widetilde B$  tend to the identity matrices with an exponential rate except the one on $(-\infty, 0)$, which leads to the following global parametrix $\widetilde P^{(\infty)}$.
\subsubsection*{RH problem for $\widetilde P^{(\infty)}$}
\begin{description}
  \item(a)  $\widetilde P^{(\infty)}(\zeta)$ is a $2\times 2$  matrix-valued function, which is analytic for $\zeta$ in
  $\mathbb{C}\backslash (-\infty,0]$.
  \item(b)  $\widetilde P^{(\infty)}(\zeta)$  satisfies the jump condition
  \begin{equation}
  \widetilde P_+^{(\infty)}(\zeta)=\widetilde P_-^{(\infty)}(\zeta)\begin{pmatrix}
                                 0 & 1 \\
                                 -1 & 0
                                 \end{pmatrix},  \qquad \zeta \in (-\infty,0).
  \end{equation}
\item(c)   As $\zeta\rightarrow \infty$, we have
\begin{equation}
 \widetilde P^{(\infty)}(\zeta) =
       (I+O(1/\zeta))  \zeta^{-\frac{1}{4}\sigma_3}\frac{I+i\sigma_1}{\sqrt{2}}.
\end{equation}
\end{description}
The solution to the above RH problem is explicitly given by
\begin{equation}\label{P-infty-omega-0}
\widetilde P^{(\infty)}(\zeta)= \zeta^{-\sigma_3/4} \frac{1}{\sqrt{2}} \begin{pmatrix}
    1 & i \\ i & 1
  \end{pmatrix}.
  \end{equation}
Near $\zeta=0$, $\widetilde B$ is approximated by the local parametrix
\begin{equation}\label{def:Bes0}
\widetilde P^{(0)}(\zeta)=E_{\alpha}(\zeta) \Phi_{2\alpha}^{(\mathrm{Bes})} \left((-s)^{3} h(\zeta)^2\right)e^{-|s|^{3/2}h(\zeta)\sigma_3},
\end{equation}
where the prefactor
\begin{equation} \label{E-z-defn}
  E_{\alpha}(\zeta) = \left(\pi (-s)^{3/2} \left(1-\frac{2\zeta}{3} \right) \right)^{\sigma_3/2}
\end{equation}
is analytic near the origin, $h$ is given in \eqref{h-function}, and $\Phi_{\alpha}^{(\Bes)}$ is the Bessel parametrix shown in Section \ref{sec:Bessel} below.

  \item The final transformation $\widetilde B \to D_{\alpha}$ is defined by
\begin{equation}
\label{D-omega-0}
 D_{\alpha}(\zeta) = \left\{
  \begin{array}{ll}
    \widetilde B(\zeta)\left( \widetilde P^{(\infty)}(\zeta)\right)^{-1}, &\quad  \hbox{$\zeta \in  \mathbb{C}\setminus \overline{U(0,\delta)}$,} \\
    \widetilde B(\zeta)\left( \widetilde P^{(0)}(\zeta)\right)^{-1}, &\quad  \hbox{$\zeta \in U(0,\delta)$.}
  \end{array}
\right.
\end{equation}
Then, $D_\alpha$ satisfies an RH problem which is analytic in $\mathbb{C}\setminus \Sigma_{D_{\alpha}}$, where $$\Sigma_{D_{\alpha}}:=\partial U(0,\delta)\cup e^{\frac{2\pi i}{3}}(\delta,+\infty)
\cup e^{-\frac{2\pi i}{3}}(\delta,+\infty).$$
Since $D_{\alpha,+}(\zeta)=D_{\alpha,-}(\zeta)J_{D_\alpha}(\zeta)$ for $\zeta \in \Sigma_{D_{\alpha}}$ with
\begin{equation}\label{eq:estjump}
J_{D_\alpha}=I+O(|s|^{-3/2}),\qquad s\to -\infty,
\end{equation}
uniformly on $\Sigma_{D_{\alpha}}$, and $$D_\alpha(\zeta)=I+O(\zeta^{-1}),\qquad  \zeta \to \infty,$$
it follows that
\begin{equation}\label{D-0-est}
  D_\alpha(\zeta) = I+O(|s|^{-3/2}), \qquad s\to -\infty,
\end{equation}
uniformly for $\zeta \in \mathbb{C} \setminus \Sigma_{D_{\alpha}}$; see the equation after \cite[Equation (A.8)]{wxz}.
\end{itemize}

\subsubsection*{The case $\alpha>-1/2$ and $\omega=e^{-2\beta \pi i}>0$}
Following \cite{ikj2009}, the analysis consists of a series explicit and invertible transformations:
\begin{equation}
\Psi \to \widetilde A \to \widehat B \to \widehat C \to D_{\alpha,\beta}.
\end{equation}
More precisely, we have
\begin{itemize}
  \item The first transformation $\Psi \to \widetilde A$ is the same as in the previous case; see \eqref{A-0}.
  \item The second transformation $\widetilde A \to \widehat B$ is defined by
 \begin{equation}\label{B-ab}
\widehat B(\zeta)=
\begin{pmatrix}
1 & 0 \\ -\frac{i|s|^{3/2}}{4} & 1
  \end{pmatrix} \widetilde A(\zeta)e^{|s|^{3/2}\widehat{g}(\zeta)\sigma_3},
\end{equation}
where
\begin{equation}\label{g-hat-function}
  \widehat{g}(\zeta) = \frac{2}{3}(\zeta-1)^{3/2}, \qquad  -\pi<\arg(\zeta-1)<\pi,
\end{equation}
so that the asymptotic behavior of $\widetilde A$ at infinity is normalized.
  \item The third transformation $\widehat B \to \widehat C$ involves lenses-opening around the interval $(0,1)$, which is based on a classical factorization of the jump matrix of $\widehat B$ on $(0,1)$. In particular, we have
     \begin{equation}\label{def:BtoC}
     \widehat C(\zeta)=\widehat B(\zeta), \qquad \textrm{$\zeta$ outside the lenses around $(0,1)$.}
     \end{equation}
As a consequence, all the jump matrices of $\widehat C$ tend to the identity matrices exponentially fast for large negative $s$ except the ones on the intervals $(-\infty,0)$ and $(0,1)$. The global parametrix in this case reads as follows.
 \subsubsection*{RH problem for $\widehat P^{(\infty)}$}
\begin{description}
  \item(a)  $\widehat P^{(\infty)}(\zeta)$ is a $2\times 2$  matrix-valued function, which is analytic for $\zeta$ in
  $\mathbb{C}\backslash (-\infty,1]$.
  \item(b)  $\widehat P^{(\infty)}(\zeta)$  satisfies the jump condition
  \begin{equation}
  \widehat P_+^{(\infty)}(\zeta)=\widehat P_-^{(\infty)}(\zeta)
  \left\{
    \begin{array}{ll}
      \begin{pmatrix}
       0 & 1
       \\
       -1 & 0
      \end{pmatrix}, & \quad \hbox{$\zeta \in (-\infty,0)$,} \\
       \begin{pmatrix}
       0 & \omega
       \\
       -\omega^{-1} & 0
      \end{pmatrix}, & \quad \hbox{$\zeta \in (0,1)$.}
    \end{array}
  \right.
  \end{equation}
\item(c)   As $\zeta\rightarrow \infty$, we have
\begin{equation}
 \widehat P^{(\infty)}(\zeta) =
       (I+O(1/\zeta))  \zeta^{-\frac{1}{4}\sigma_3}\frac{I+i\sigma_1}{\sqrt{2}}.
\end{equation}
\end{description}
The solution to the above RH problem is explicitly given by
\begin{equation}
\widehat P^{(\infty)}(\zeta)=
\begin{pmatrix}
1 & 0
\\
2\omega & 1
\end{pmatrix}(\zeta-1)^{-\sigma_3/4} \frac{1}{\sqrt{2}} \begin{pmatrix}
    1 & i \\ i & 1
  \end{pmatrix}
\left(\frac{(\zeta-1)^{1/2}+i}{(\zeta-1)^{1/2}-i}\right)^{\omega \sigma_3},
  \end{equation}
where the branches of the arguments are fixed by the inequalities
$$-\pi <\arg(\zeta-1)<\pi, \qquad -\pi<\arg\left(\frac{(\zeta-1)^{1/2}+i}{(\zeta-1)^{1/2}-i}\right)<\pi;$$
see \cite[Equation (A.24)]{ikj2009}. Near $\zeta=0$, $\widehat C$ is approximated by the local parametrix
\begin{equation}\label{def:CHFpara}
\widehat P^{(0)}(\zeta)=E_{\alpha,\beta}(\zeta) \Phi_{\alpha,\beta}^{(\mathrm{CHF})} \left((-s)^{3/2} f(\zeta)\right)e^{|s|^{3/2}\widehat{g}(\zeta)\sigma_3},
\end{equation}
where the prefactor $E_{\alpha,\beta}(\zeta)$ is analytic at $\zeta =0$ with
\begin{equation} \label{E-alpha-beta-0}
  E_{\alpha,\beta}(0) = \frac{1}{\sqrt{2}} \begin{pmatrix}
    1 & 1 \\ 2 \beta - 1 & 2 \beta + 1
  \end{pmatrix} (4(-s)^{3/2})^{\beta \sigma_3} e^{[\frac{2i(-s)^{3/2}}{3} - i \pi (\frac{1}{4} + \alpha - \frac{\beta}{2})] \sigma_3},
\end{equation}
$f$ is given by
\begin{equation}\label{def:fzeta}
  f(\zeta) = \frac{2}{3} - \frac{2i}{3} (\zeta-1)^{3/2}, \qquad 0<\arg(\zeta-1)<2\pi,
\end{equation}
and $\Phi_{\alpha,\beta}^{(\mathrm{CHF})}$ is the confluent hypergeometric parametrix shown in Section \ref{sec:CHF} below;
see \cite[Equation (A.34)]{ikj2009}. Near $\zeta=1$, the local parametrix $\widehat P^{(1)}$ is constructed with the aid of Airy parametrix $\Phi^{(\Ai)}$, similar to the local parametrix built in \eqref{eq:Airypara}, we omit the details here.
  \item The final transformation $\widehat C \to D_{\alpha,\beta}$ is defined by
  \begin{equation} \label{D-ab}
 D_{\alpha,\beta}(\zeta) =
\left\{
  \begin{array}{ll}
    \widehat C(\zeta)\left(\widehat P^{(\infty)}(\zeta)\right)^{-1}, & \quad \hbox{$\zeta \in  \mathbb{C}\setminus (\overline{U(0,\delta)}\cup \overline{U(1,\delta)})$,} \\
    \widehat C(\zeta)\left(\widehat P^{(0)}(\zeta)\right)^{-1}, & \quad \hbox{$\zeta \in U(0,\delta)$,} \\
    \widehat C(\zeta)\left(\widehat P^{(1)}(\zeta)\right)^{-1}, & \quad \hbox{$\zeta \in U(1,\delta)$.}
  \end{array}
\right.
\end{equation}
Then, $D_{\alpha,\beta}$ satisfies an RH problem which is analytic in $\mathbb{C}\setminus \Sigma_{D_{\alpha,\beta}}$, where
\begin{align}
\Sigma_{D_{\alpha,\beta}}& :=   \partial U(0,\delta)\cup e^{\frac{2\pi i}{3}}(\delta,+\infty)
\cup e^{-\frac{2\pi i}{3}}(\delta,+\infty)\cup \partial U(1,\delta)\cup(1+\delta,+\infty)
\nonumber
\\
& \quad ~~\cup \textrm{\{the boundaries of the lenses outside $U(0,\delta)$ and $U(1,\delta)\}$}.
\end{align}
Since $D_{\alpha,\beta}$ tends to the identity matrix at infinity, and its jump matrices still satisfy the estimate \eqref{eq:estjump} uniformly on the jump contour $\Sigma_{D_{\alpha,\beta}}$, we finally conclude that
\begin{equation}\label{D-ab-est}
  D_{\alpha,\beta}(\zeta) = I+O(|s|^{-3/2}), \quad s\to -\infty,
\end{equation}
uniformly for $\zeta \in \mathbb{C} \setminus \Sigma_{D_{\alpha,\beta}}$; see \cite[Equation (3.35)]{ikj2009}.
\end{itemize}

We are now ready to derive the asymptotics of $w,H$ and relevant functions for large negative $s$.

\subsection{Asymptotics of $w,H$ and relevant functions for large negative $s$}
Since the asymptotic analysis of the RH problem for $\Psi$ is different for $\omega=0$
and $\omega>0$, accordingly, we separate the asymptotic results into two propositions below. We start with the case $\omega=0$.

\begin{pro}\label{Pro: asy phi negative infty}
For $\alpha>-\frac 12$ and $\omega =0$, we have the following asymptotics as $s\to -\infty$:
\begin{align}
  u(s;2\alpha,0) & = -\frac{s}{2}+(2\alpha^2-\frac{1}{8})\frac{1}{s^2}+O(|s|^{-7/2}), \label{eq:asy-u-negative infty} \\
  w(s;2\alpha+\frac {1}{2}, 0)& =(2\alpha+\frac{1}{2})\frac {1}{s}+O(s^{-9/2}),   \label{eq:asy-w-negative infty} \\
  H(s;2\alpha,0)& =
    \frac{s^2}{4}+(2\alpha^2-\frac{1}{8})\frac{1}{s}+O(|s|^{-5/2}),  \label{eq:asy-H-negative infty} \\
  \ln\left|\left(\Phi_0\right)_{11}(s)\right|&=(2\alpha+\frac{1}{2})\ln|s|- 2\alpha\ln 2-\ln \Gamma(1+2\alpha)
 + \frac{1}{2}\ln \pi  +O(|s|^{-3/2}). \label{eq:asy-phi-0-negative infty}
\end{align}
\end{pro}
\begin{proof}
The asymptotics for the Painlev\'e XXXIV transcendent $u(s;2\alpha,0)$ and the Hamiltonian $H(s;2\alpha,0)$ in \eqref{eq:asy-u-negative infty} and \eqref{eq:asy-H-negative infty} are proved in Theorem 2 and Theorem 1 of \cite{wxz}, respectively. The asymptotics of the tronqu\'ee solution $w(s;2\alpha+\frac {1}{2}, 0)$ in \eqref{eq:asy-w-negative infty}  is obtained in \cite[Equations (19) and (31)]{Kapaev92}.
The leading term in the expansion \eqref{eq:asy-w-negative infty}  can also be derived by using the second equation of  \eqref{eq:Hamiltonian system} and  the expansion \eqref{eq:asy-u-negative infty}.


To show the asymptotics of $\ln |\left(\Phi_0\right)_{11}(s)|$,  as in the derivation of \eqref{eq:asy-phi-0}, we explore the explicit expression of $\Psi(-s\zeta; s)$ in $\Omega_2$. 
By inverting the transformations \eqref{A-0}, \eqref{B-0} and \eqref{D-omega-0}, we obtain from \eqref{def:Bes0} that
\begin{align} \label{Phi-Bessel-expression}
  \Psi(-s\zeta; s) =
\begin{pmatrix}
1 & 0
\\
ia(s) & 1
\end{pmatrix}(-s)^{-\sigma_3/4}D_{\alpha}(\zeta) E_{\alpha}(\zeta) \Phi_{2\alpha}^{(\mathrm{Bes})} \left((-s)^{3} h(\zeta)^2 \right), \qquad \zeta \in \Omega_2.
\end{align}

  \begin{figure}[t]
\begin{center}
   \setlength{\unitlength}{1truemm}
   \begin{picture}(80,70)(-5,2)
       \put(40,40){\line(-2,-3){16}}
       \put(40,40){\line(-2,3){16}}
       \put(40,40){\line(-1,0){30}}

       \put(30,55){\thicklines\vector(2,-3){1}}
       \put(30,40){\thicklines\vector(1,0){1}}
       \put(30,25){\thicklines\vector(2,3){1}}

       \put(39,36.3){$0$}
       \put(20,11){$\Sigma_3$}
       \put(20,69){$\Sigma_1$}
       \put(3,40){$\Sigma_2$}

       \put(52,39){$\texttt{I}$}
       \put(25,44){$\texttt{II}$}
       \put(25,34){$\texttt{III}$}

       \put(40,40){\thicklines\circle*{1}}

   \end{picture}
   \caption{The  regions in the $z$ plane for $\Phi^{(\Bes)}_{\alpha}$.}
   \label{fig:regions-Phi-B}
\end{center}
\end{figure}
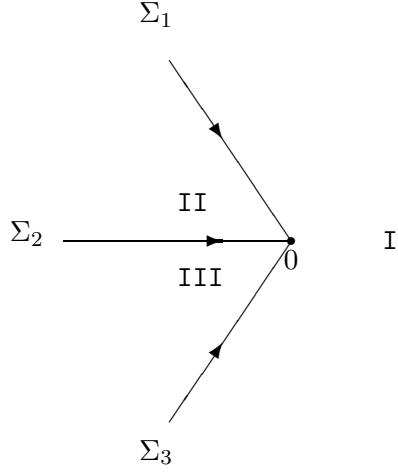

By \eqref{h-function}, it follows that
$$ h(\zeta)^2 \sim \zeta, \qquad \zeta \to 0,$$
and
$$z = (-s)^{3} h(\zeta)^2 \in \texttt{II}, \qquad \zeta\in\Omega_2, $$
where the region $\texttt{II}$ is shown in Figure \ref{fig:regions-Phi-B}. In addition, from the explicit construction of the Bessel parametrix $\Phi_{\alpha}^{(\mathrm{Bes})}$ given in \eqref{Phi-B-solution}, we have
\begin{align}\label{eq:phibes}
  \Phi_{2\alpha}^{(\mathrm{Bes})}(z) & = \begin{pmatrix}
    I_{2\alpha}(z^{1/2}) & \frac{i}{\pi} K_{2\alpha}(z^{1/2}) \\
    \pi i z^{1/2} \, I_{2\alpha}'(z^{1/2}) & -z^{1/2} \, K_{2\alpha}'(z^{\frac{1}{2}})
  \end{pmatrix} J_2^{-1},  \qquad z \in \texttt{II}.
\end{align}
We now need to know the asymptotic behavior of $\Phi_{2\alpha}^{(\mathrm{Bes})} \left((-s)^{3} h(\zeta)^2 \right)$ as $\zeta \to 0$ from $\Omega_2$. Recall the following asymptotics of $I_{2\alpha}(z)$ and $K_{2\alpha} (z)$ near the origin:
\begin{align}
  I_{2\alpha}(z) & \sim \frac{z^{2\alpha}}{2^{2\alpha} \Gamma(2\alpha + 1)}, 
  \\
   K_{2\alpha}(z) & \sim \frac{\pi}{2 \sin(2\alpha \pi)} \begin{cases}
    \frac{(z/2)^{-2\alpha}}{\Gamma(-2\alpha + 1)}, &\quad \alpha > 0, 2\alpha \notin \mathbb{N}, \\
    \frac{(z/2)^{2\alpha}}{\Gamma(2\alpha + 1)}, &\quad  -\frac{1}{2}< \alpha < 0,
  \end{cases} \\
  K_{0}(z) & \sim - \ln z, \qquad   K_{n}(z)  \sim \frac{\Gamma(n)}{2^{n+1} z^{n}};
\end{align}
see \cite[\S 10.31]{DLMF}. This, together with \eqref{eq:phibes}, implies that, as $z \to 0$,
\begin{align}
  \Phi_{2\alpha}^{(\mathrm{Bes})}(z) J_2 &\sim \begin{pmatrix}
    \frac{z^\alpha}{2^{2\alpha} \Gamma(2\alpha + 1)} & \frac{i}{2 \sin(2 \alpha \pi)} \frac{2^{2\alpha}}{\Gamma(-2\alpha + 1) z^{\alpha}} \vspace{2mm} \\
    \frac{\pi i  z^\alpha}{2^{2\alpha} \Gamma(2\alpha )} & - \frac{\pi}{2 \sin(2 \alpha \pi)} \frac{2^{2\alpha}}{\Gamma(-2\alpha ) z^{\alpha}}
  \end{pmatrix}, \qquad \alpha > 0, 2\alpha \notin \mathbb{N}, \label{Bessel-asy-1} \\
  \Phi_{2\alpha}^{(\mathrm{Bes})}(z) J_2 &\sim \begin{pmatrix}
    \frac{z^\alpha}{2^{2\alpha} \Gamma(2\alpha + 1)} & -\frac{i}{2 \sin(2 \alpha \pi)} \frac{z^{\alpha}}{2^{2\alpha}\Gamma(2\alpha + 1) } \vspace{2mm} \\
    \frac{\pi i  z^\alpha}{2^{2\alpha} \Gamma(2\alpha )} & \frac{\pi}{2 \sin(2 \alpha \pi)} \frac{z^{\alpha}}{2^{2\alpha}\Gamma(2\alpha ) }
  \end{pmatrix}, \qquad -\frac{1}{2} < \alpha < 0, \label{Bessel-asy-2}
\end{align}
\begin{align}
  \Phi_{0}^{(\mathrm{Bes})}(z) J_2 &\sim \begin{pmatrix}
    1 & -\frac{i}{2\pi} \ln z \vspace{2mm} \\
    \frac{\pi i }{2} z & 1
  \end{pmatrix},  \label{Bessel-asy-3} \\
  \Phi_{2\alpha}^{(\mathrm{Bes})}(z) J_2 &\sim \begin{pmatrix}
    \frac{z^\alpha}{2^{2\alpha} \Gamma(2\alpha + 1)} & \frac{i}{\pi} \frac{\Gamma(2\alpha )}{2^{2\alpha+1} z^{\alpha}} \vspace{2mm} \\
    \frac{\pi i  z^\alpha}{2^{2\alpha} \Gamma(2\alpha )} & \frac{\Gamma(2\alpha +1)}{2^{2\alpha+1} z^{\alpha}}
  \end{pmatrix}, \qquad 2\alpha \in \mathbb{N}. \label{Bessel-asy-4}
\end{align}

If $2\alpha\notin \mathbb{N}_0$ and $\alpha>-\frac {1}{2}$, we see from \eqref{Phi0-limit-relation}, \eqref{E-z-defn} and \eqref{Phi-Bessel-expression} that
\begin{multline}
  \Phi_0(s)  = \begin{pmatrix}
1 & 0
\\
ia(s) & 1
\end{pmatrix} (-s)^{-\sigma_3/4}D_\alpha(0)(\pi (-s)^{3/2})^{\sigma_3/2}
 \\
\times \displaystyle \lim_{\zeta \to 0 , \atop \zeta \in \Omega_2} \left[ \Phi_{\alpha}^{(\mathrm{Bes})} \left((-s)^{3} h(\zeta) \right) E_2^{-1}  \zeta^{-\alpha \sigma_3} \right]  (-s)^{-\alpha \sigma_3}. \label{prop2-Phi0-limit-alpha-1}
\end{multline}
In view of \eqref{E2-J2-1} with $\omega=0$, it is readily seen from the above formula and \eqref{Bessel-asy-1} that
\begin{align} \label{prop2-Phi0-limit-alpha-2}
  \Phi_0(s) & = \begin{pmatrix}
1 & 0
\\
ia(s) & 1
\end{pmatrix} (-s)^{-\sigma_3/4} D_\alpha(0) (\pi (-s)^{3/2})^{\sigma_3/2}
\frac{e^{-\pi i \alpha} (-s)^{3\alpha}}{2^{2\alpha} \Gamma(2\alpha + 1)}  \begin{pmatrix}
    1 & * \\ 2\alpha \pi i & *
  \end{pmatrix}  (-s)^{-\alpha \sigma_3},
\end{align}
for $\alpha>0$ and $2\alpha \notin \mathbb{N}$. For the case $-\frac{1}{2} < \alpha < 0$, one obtains the above limit as well with the aid of  \eqref{Bessel-asy-2} and  \eqref{prop2-Phi0-limit-alpha-1}.

If $2\alpha \in \mathbb{N}_0$, as in the proof of \eqref{eq:asy-phi-0}, we obtain from \eqref{Phi0-limit-relation}, \eqref{E-z-defn} and \eqref{Phi-Bessel-expression} that the formula \eqref{prop2-Phi0-limit-alpha-2} still holds if the constant $e^{-\pi i \alpha}$ is replaced by $1$.

Finally, by letting $s\to -\infty$, we obtain the asymptotics for $\ln|\left(\Phi_0\right)_{11}(s)|$ in \eqref{eq:asy-phi-0-negative infty} from \eqref{prop2-Phi0-limit-alpha-2} and \eqref{D-0-est}.

This completes the proof of Proposition \ref{Pro: asy phi negative infty}.
\end{proof}

Next, we derive the asymptotics when $\omega > 0$.
\begin{pro}\label{Pro: asy phi negative infty-omega}
For $\alpha>-\frac 12$ and $\omega = e^{-2\beta \pi i}$ with $\beta i \in \mathbb{R}$, we have the following asymptotics as $s\to -\infty$:
\begin{align}
  u(s;2\alpha,\omega) & = \frac{2|\alpha-\beta|}{\sqrt{|s|}} \cos\left(\frac{\theta(s)}{2}+\arg\Gamma(1+\alpha-\beta)-\frac{\pi}{4}\right) \nonumber \\
    & \quad  \times \cos\left(\frac{\theta(s)}{2}+\arg\Gamma(\alpha-\beta)+\frac{\pi}{4}\right) +O(s^{-2}),   \label{eq:asy-u-negative infty-omega} \\
  w(s;2\alpha+\frac {1}{2},\omega)& =  |s|^{1/2}  \tan \left( \frac{\theta(s)}{2}  + \arg \Gamma(1+\alpha - \beta) -
    \frac{\pi}{4} \right)  + O(s^{-1}),    \label{eq:asy-w-negative infty-omega}
  \\
  H(s;2\alpha,\omega)   &= 2 \beta i |s|^{1/2}  - \frac{|\alpha-\beta|}{2s} \sin\left(\theta(s) +2\arg\Gamma(\alpha-\beta)+\arg(\alpha-\beta) \right)  \nonumber \\
   & \quad + \frac{\alpha^2 - 3\beta^2 }{2s}  +O(|s|^{ - 5/2}), \label{eq:asy-H-negative infty-omega}
\\
\ln\left|\left(\Phi_0\right)_{11}(s)\right| &= (\frac{\alpha}{2} - \frac{1}{4})\ln(|s|)+\ln\left(\left|\cos \left( \frac{1}{2} \theta(s) + \arg \Gamma(1+\alpha - \beta) - \frac{\pi}{4} \right) \right|\right)\nonumber\\
&\quad +\ln\left(\frac{|\Gamma(1+\alpha-\beta)|}{\Gamma(1+2\alpha)}\right)-\frac{\beta}{2}\pi i +(\alpha+\frac{1}{2})\ln 2+O(|s|^{-3/2}), \label{eq:asy-phi-0-negative infty-omega}
\end{align}
where $\theta(s):=\theta(s;\alpha,\beta)$ is given by
\begin{equation} \label{theta}
  \theta(s;\alpha,\beta) = \frac{4}{3} |s|^{3/2} - \alpha \pi - 6 i \beta \ln 2 - 3i \beta \ln|s|.
\end{equation}
\end{pro}
\begin{proof}
Again, the asymptotics for the Painlev\'e XXXIV transcendent $u(s;2\alpha,\omega)$ in \eqref{eq:asy-u-negative infty-omega} and the Hamiltonian $H(s;2\alpha,\omega)$ in \eqref{eq:asy-H-negative infty-omega} are given in Theorem 2 and Theorem 1 of \cite{wxz}, respectively.  The asymptotics of the tronqu\'ee solution $w(s;2\alpha+\frac {1}{2}, \omega)$ in  \eqref{eq:asy-w-negative infty}   is given in \cite[Equations (19) and (28)]{Kapaev92}.  The leading term in the expansion \eqref{eq:asy-w-negative infty}  can also be derived by using the second equation of  \eqref{eq:Hamiltonian system} and  the expansion  \eqref{eq:asy-u-negative infty-omega}.


To derive the asymptotics of $\ln|\left(\Phi_0\right)_{11}(s)|$, similar to the case $\omega=0$ in the previous proposition, we consider the explicit expression of $\Psi(-s\zeta; s)$ in $\Omega_2$. 
By inverting the transformations \eqref{A-0}, \eqref{B-ab}, \eqref{def:BtoC} and \eqref{D-ab}, it follows from \eqref{def:CHFpara} that
  \begin{multline} \label{Psi-neg-sz-in-II}
  \Psi(-s\zeta; s)  = \begin{pmatrix}
1 & 0
\\
ia(s) & 1
\end{pmatrix}(-s)^{-\sigma_3/4} \begin{pmatrix}
    1 & 0 \\ \frac{i(-s)^{3/2}}{4} & 1
  \end{pmatrix}
\\ \times
 D_{\alpha,\beta} (\zeta) E_{\alpha,\beta}(\zeta) \Phi_{\alpha,\beta}^{(\mathrm{CHF})} \left((-s)^{3/2} f(\zeta)\right), \qquad \zeta \in \Omega_2.
\end{multline}
On account of the definition of $f$ given in \eqref{def:fzeta}, it is easily seen that
\begin{equation}\label{eq:fzero}
 f(\zeta) \sim \zeta,  \qquad \zeta \to 0,
\end{equation}
and
$$z = (-s)^{3/2} f(\zeta) \in \texttt{IV}, \qquad \zeta \in \Omega_2,$$
where the region $\texttt{IV}$ in this case is illustrated in Figure \ref{fig:regions-Phi}.
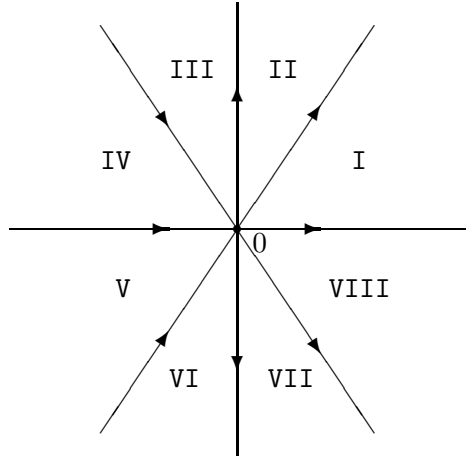
\begin{figure}[t]
\begin{center}
   \setlength{\unitlength}{1truemm}
   \begin{picture}(100,70)(-5,2)
       \put(40,40){\line(-2,-3){18}}
       \put(40,40){\line(-2,3){18}}
       \put(40,40){\line(-1,0){30}}
       \put(40,40){\line(1,0){30}}
      \put(40,40){\line(0,1){30}}
 \put(40,40){\line(0,-1){30}}
  \put(40,40){\line(2,-3){18}}
    \put(40,40){\line(2,3){18}}

       \put(30,55){\thicklines\vector(2,-3){1}}
       \put(30,40){\thicklines\vector(1,0){1}}
       \put(50,40){\thicklines\vector(1,0){1}}
       \put(30,25){\thicklines\vector(2,3){1}}
      \put(50,25){\thicklines\vector(2,-3){1}}
       \put(50,55){\thicklines\vector(2,3){1}}

     \put(40,58){\thicklines\vector(0,1){1}}
      \put(40,22){\thicklines\vector(0,-1){1}}

       \put(42,36.9){$0$}

          \put(55,48){$\texttt{I}$}
           \put(44,60){$\texttt{II}$}
              \put(31,60){$\texttt{III}$}
            \put(22,48){$\texttt{IV}$}
        \put(24,31){$\texttt{V}$}
        \put(31,19){$\texttt{VI}$}
         \put(44,19){$\texttt{VII}$}
 \put(52,31){$\texttt{VIII}$}

       \put(40,40){\thicklines\circle*{1}}
\end{picture}
   \caption{The  regions $\texttt{I}-\texttt{VIII}$ in the $z$ plane for $\Phi_{\alpha,\beta}^{(\mathrm{CHF})}$.}
   \label{fig:regions-Phi}
\end{center}
\end{figure}

The explicit expression of $\Phi_{\alpha,\beta}^{(\mathrm{CHF})}(z)$ for $z\in\texttt{IV}$ is given by
\begin{multline}
  \Phi_{\alpha,\beta}^{(\mathrm{CHF})}(z)  = 2^{\beta \sigma_3} e^{\frac{\pi i \beta}{2}  \sigma_3} \begin{pmatrix}
    e^{- \pi i (\alpha + 2\beta)} & 0 \\ 0 & e^{\pi i (2\alpha + \beta)}
  \end{pmatrix}
  M(z)
\\
\times \begin{pmatrix}
    e^{\pi i \alpha} & 0 \\ 2i \sin(\pi (\beta - \alpha)) & e^{- \pi i \alpha}
  \end{pmatrix} e^{\frac{\pi i \beta}{2} \sigma_3},  \label{Phi-CHF-expression}
\end{multline}
where
\begin{align}
  & M(z) = \left(\begin{array}{l}
    (2 e^{\frac{\pi i }{2}} z)^\alpha \, \psi(\alpha + \beta, 1+2\alpha, 2e^{\frac{\pi i}{2}} z ) e^{\pi i (\alpha + 2\beta )} e^{-iz} \vspace{2mm} \\
    - \frac{\Gamma(1+\alpha + \beta)}{\Gamma(\alpha - \beta)} (2 e^{\frac{\pi i }{2}} z)^{-\alpha}  \, \psi(1-\alpha + \beta, 1-2\alpha, 2e^{\frac{\pi i}{2}} z ) e^{\pi i (\beta - 3 \alpha)} e^{-iz}
  \end{array} \right. \nonumber \\
  &  \left. \qquad \qquad \qquad \begin{array}{l}
    - \frac{\Gamma(1+\alpha - \beta)}{\Gamma(\alpha + \beta)} (2 e^{-\frac{\pi i }{2}} z)^\alpha \, \psi(1+\alpha - \beta, 1+2\alpha, 2e^{-\frac{\pi i}{2}} z ) e^{\pi i (2\alpha + \beta )} e^{iz} \vspace{2mm} \\
    (2 e^{-\frac{\pi i }{2}} z)^{-\alpha} \, \psi(-\alpha - \beta, 1-2\alpha, 2e^{-\frac{\pi i}{2}} z ) e^{ - 2\alpha \pi i } e^{iz}
  \end{array}\right) \label{eq:M-fun-def}
\end{align}
with branch cuts chosen such that $-\frac{\pi }{2} < \arg z < \frac{3\pi}{2}$; see \eqref{Phi-C-solution} below. The confluent hypergeometric function $\psi(a,b,z)$ is the unique solution to the Kummer's equation
$$z\frac{d^2y}{dz^2}+(b-z)\frac{dy}{dz}-ay=0$$
satisfying the boundary condition $\psi(a,b,z)\sim  z^{-a}$ as $z\to \infty$ and $-\frac{3\pi }{2} < \arg z < \frac{3\pi}{2}$; see \cite[Chapter 13]{DLMF}.

If $2\alpha\notin \mathbb{N}_0$ and $\alpha>-\frac {1}{2}$, we obtain from \eqref{Phi0-limit-relation},
\eqref{Psi-neg-sz-in-II} and \eqref{Phi-CHF-expression} that
\begin{align} \label{Phi0-limit-(-s)}
  \Phi_0(s) & = \begin{pmatrix}
1 & 0
\\
ia(s) & 1
\end{pmatrix} (-s)^{-\sigma_3/4} \begin{pmatrix}
    1 & 0 \\ \frac{i(-s)^{3/2}}{4} & 1
  \end{pmatrix} D_{\alpha,\beta}(0) E_{\alpha,\beta}(0) 2^{\beta \sigma_3} e^{\frac{\pi i \beta}{2} \sigma_3} \begin{pmatrix}
    e^{- \pi i (\alpha + 2\beta)} & 0 \\ 0 & e^{\pi i (2\alpha + \beta)}
  \end{pmatrix} \nonumber \\
  & \quad \times   \displaystyle \lim_{\zeta \to 0 , \atop \zeta \in \Omega_2} \left[ M((-s)^{3/2} f(\zeta)) \begin{pmatrix}
    e^{\pi i \alpha} & 0 \\ 2i \sin(\pi (\beta - \alpha)) & e^{- \pi i \alpha}
  \end{pmatrix} e^{\frac{\pi i \beta}{2} \sigma_3} E_2^{-1}  \zeta^{-\alpha \sigma_3} \right]
\nonumber
\\
& \quad \times (-s)^{-\alpha \sigma_3}.
\end{align}
To study the above limit, we start with some properties of the confluent hypergeometric function $\psi(a,b,z)$ listed below:
\begin{align}
  \psi(a,b,z) & =  \frac{\Gamma(1-b)}{\Gamma(a-b+1)} \phi(a,b,z)+ \frac{\Gamma(b-1)}{\Gamma(a)} \phi(a-b+1,2-b,z)z^{1-b}, \label{eq:K-relation-1} \\
  \phi(a,b,z) & = \frac{\Gamma(b)}{\Gamma(b-a)} e^{i\pi a}\psi(a,b,z)+ \frac{\Gamma(b)}{\Gamma(a)}e^{i\pi(a-b)} \phi(b-a,b,e^{-i\pi}z)e^{z}; \label{eq:K-relation-2}
\end{align}
see \cite[Equations 13.2.41 and 13.2.42]{DLMF}. In the above formulas, $\phi(a,b,z)$ is the other confluent hypergeometric function defined by
\begin{equation}\label{eq:K}
  \phi(a,b,z) = \sum_{n=0}^{\infty}\frac{(a)_n}{(b)_n}\frac{z^n}{n!},
 \end{equation}
which is an entire function in $z$ and satisfies the relation
 \begin{equation}\label{eq:K-relation-3}
  \phi(a,b,z) = \phi(b-a,b,-z)e^{z};
 \end{equation}
see \cite[Equations 13.2.2 and 13.2.39]{DLMF}. Since
\begin{equation} \label{E2-1-product-matrix-1}
  \begin{pmatrix}
    e^{\pi i \alpha} & 0 \\ 2i \sin(\pi (\beta - \alpha)) & e^{- \pi i \alpha}
  \end{pmatrix} e^{\frac{\pi i \beta}{2} \sigma_3} E_2^{-1} = \begin{pmatrix}
    e^{\frac{\beta \pi i}{2} } & * \\
    e^{(\frac{3 \beta}{2} - 2 \alpha) \pi i } & *
  \end{pmatrix},
\end{equation}
it follows from \eqref{eq:M-fun-def} and \eqref{eq:K-relation-2} that
\begin{align}
  & \left[ M(z) \begin{pmatrix}
    e^{\pi i \alpha} & 0 \\ 2i \sin(\pi (\beta - \alpha)) & e^{- \pi i \alpha}
  \end{pmatrix} e^{\frac{\pi i \beta}{2} \sigma_3} E_2^{-1}  \right]_{11} \nonumber \\
  &  = (2 e^{\frac{\pi i }{2}} z)^\alpha \, \psi(\alpha + \beta, 1+2\alpha, 2e^{\frac{\pi i}{2}} z ) e^{\pi i (\alpha + 2\beta )} e^{-iz}  e^{\frac{\beta \pi i}{2} }  \nonumber \\
  &  \quad - \frac{\Gamma(1+\alpha - \beta)}{\Gamma(\alpha + \beta)} (2 e^{-\frac{\pi i }{2}} z)^\alpha \, \psi(1+\alpha - \beta, 1+2\alpha, 2e^{-\frac{\pi i}{2}} z ) e^{\pi i (2\alpha + \beta )} e^{iz}  e^{(\frac{3 \beta}{2} - 2 \alpha) \pi i } \nonumber \\
  &  = 2^{\alpha}e^{(\frac{ \alpha}{2} + \frac{3 \beta}{2} )  \pi i }\frac{\Gamma(1+\alpha - \beta)}{\Gamma(1+2\alpha)}\phi(\alpha+\beta,1+2\alpha,2e^{\frac{\pi i}{2}} z )e^{-iz}z^{\alpha}. \label{eq:M-11-entry}
\end{align}
We next evaluate $(2,1)$ entry of the matrix $M(z) \begin{pmatrix}
    e^{\pi i \alpha} & 0 \\ 2i \sin(\pi (\beta - \alpha)) & e^{- \pi i \alpha}
  \end{pmatrix} e^{\frac{\pi i \beta}{2} \sigma_3} E_2^{-1}$. From  \eqref{eq:K-relation-1}  and  \eqref{eq:K-relation-3}, it is easily seen that
\begin{align*}
  \psi(1-\alpha + \beta, 1-2\alpha, 2e^{\frac{\pi i}{2}} z ) & = \frac{\Gamma(2\alpha)}{\Gamma(1+\alpha+\beta)} \phi(1-\alpha+\beta, 1-2\alpha, 2e^{\frac{\pi i}{2}} z) \nonumber \\
   & \quad + \frac{\Gamma(-2\alpha)}{\Gamma(1-\alpha+\beta)} \phi(1+\alpha+\beta, 1+2\alpha, 2e^{\frac{\pi i}{2}} z) \, ( 2e^{\frac{\pi i}{2}} z)^{2 \alpha}, \\
   \psi(-\alpha - \beta, 1-2\alpha, 2e^{-\frac{\pi i}{2}} z ) & = \frac{\Gamma(2\alpha)}{\Gamma(\alpha-\beta)} \phi(1-\alpha+\beta, 1-2\alpha, 2e^{\frac{\pi i}{2}} z) \, e^{2iz} \nonumber \\
   & \quad + \frac{\Gamma(-2\alpha)}{\Gamma(-\alpha-\beta)} \phi(1+\alpha+\beta, 1+2\alpha, 2e^{\frac{\pi i}{2}} z) \, ( 2e^{-\frac{\pi i}{2}} z)^{2 \alpha} e^{-2iz}.
\end{align*}
This, together with the reflection formula
$$\Gamma(z) \Gamma(1-z) = \pi/ \sin(\pi z),\qquad z\neq 0, \pm 1, \pm2, \ldots,$$ implies that
\begin{align}
  & \left[ M(z) \begin{pmatrix}
    e^{\pi i \alpha} & 0 \\ 2i \sin(\pi (\beta - \alpha)) & e^{- \pi i \alpha}
  \end{pmatrix} e^{\frac{\pi i \beta}{2} \sigma_3} E_2^{-1}  \right]_{21} \nonumber \\
  &  = - \frac{\Gamma(1+\alpha + \beta)}{\Gamma(\alpha - \beta)} (2 e^{\frac{\pi i }{2}} z)^{-\alpha}  \, \psi(1-\alpha + \beta, 1-2\alpha, 2e^{\frac{\pi i}{2}} z ) e^{\pi i (\beta - 3 \alpha)} e^{-iz}  e^{\frac{\beta \pi i}{2} }  \nonumber \\
  &   \quad +(2 e^{-\frac{\pi i }{2}} z)^{-\alpha} \, \psi(-\alpha - \beta, 1-2\alpha, 2e^{-\frac{\pi i}{2}} z ) e^{ - 2\alpha \pi i } e^{iz}  e^{(\frac{3 \beta}{2} - 2 \alpha) \pi i } \nonumber \\
  &  = 2^{\alpha}e^{(-\frac{7 \alpha}{2} + \frac{\beta}{2} )  \pi i }\frac{\Gamma(1+\alpha + \beta)}{\Gamma(1+2\alpha)}\phi(1+\alpha+\beta,1+2\alpha,2e^{\frac{\pi i}{2}} z )e^{-iz}z^{\alpha}.  \label{eq:M-21-entry}
\end{align}
Inserting \eqref{eq:M-11-entry} and \eqref{eq:M-21-entry} into \eqref{Phi0-limit-(-s)}, we obtain from \eqref{E-alpha-beta-0}, \eqref{eq:fzero}, the fact (see \cite[Equations 13.2.13 and 13.2.16]{DLMF})
\begin{equation}
\phi(a,b,z ) \sim 1 , \qquad \psi(a,b,z) \sim \frac{\Gamma(b-1)}{\Gamma(a)}z^{1-b}, \qquad z\to 0,
\end{equation}
and \eqref{D-ab-est}, that
\begin{align*}
  \left(\Phi_0\right)_{11}(s) & =  (-s)^{\frac{\alpha}{2} - \frac{1}{4}}  2^{\alpha-1/2}  e^{(-\alpha + \frac{\beta}{2})\pi i } \biggl[2^{3\beta}\frac{\Gamma(1+\alpha-\beta)}{\Gamma(1+2\alpha)}|s|^{\frac{3}{2}\beta}\exp(\frac{2}{3}i|s|^{\frac{3}{2}}-\frac{1}{4}\pi i -\frac{\alpha}{2}\pi i )\nonumber\\
&\quad + 2^{-3\beta}\frac{\Gamma(1+\alpha+\beta)}{\Gamma(1+2\alpha)}|s|^{-\frac{3}{2}\beta}\exp(-\frac{2}{3}i|s|^{\frac{3}{2}}+\frac{1}{4}\pi i+\frac{\alpha}{2}\pi i)\biggr](1+O(|s|^{-3/2}))
\end{align*}
for large negative $s$. Recall that $\beta$ is purely imaginary, the above formula can be further simplified as
\begin{multline}\label{Phi0-limit-(-s)-pure I beta}
 \left(\Phi_0\right)_{11}(s)  = \frac{ |\Gamma(1+\alpha-\beta)|e^{ (- \alpha + \frac{\beta}{2})\pi i }2^{\alpha +\frac{1}{2}} (-s)^{\frac{\alpha}{2} - \frac{1}{4}} }{\Gamma(1+2\alpha)}   \\
 \times \cos\left(\frac{1}{2}\theta(s)+\arg\Gamma(1+\alpha-\beta)-\frac{\pi}{4}\right)(1+O(|s|^{-3/2})),
\end{multline}
where $\theta(s)$ is defined in \eqref{theta}. The asymptotic formula of $\ln |(\Phi_0)_{11}(s)|$ in \eqref{eq:asy-phi-0-negative infty-omega} then follows directly from \eqref{Phi0-limit-(-s)-pure I beta}.

If $2\alpha \in \mathbb{N}_0$, due to \eqref{Phi0-limit-relation}, $\Phi_0(s)$ takes a similar expression as \eqref{Phi0-limit-(-s)}, except that the limit is modified to be
\begin{equation}
  \lim_{\zeta \to 0 , \atop \zeta \in \Omega_2} \Bigg[ M((-s)^{3/2} f(\zeta))
   \begin{pmatrix}
    e^{\pi i \alpha} & 0
    \\
    2i \sin(\pi (\beta - \alpha)) & e^{- \pi i \alpha}
  \end{pmatrix}
  e^{\frac{\pi i \beta}{2} \sigma_3} E_2^{-1} \left (I-\frac{\kappa}{2\pi i}\ln \zeta ~\sigma_+\right)   \zeta^{-\alpha \sigma_3} \Bigg].
\end{equation}
In view of the expression of $E_2$ in \eqref{E2-def}, we have
\begin{equation} \label{E2-1-product-matrix-2}
  \begin{pmatrix}
    e^{\pi i \alpha} & 0 \\ 2i \sin(\pi (\beta - \alpha)) & e^{- \pi i \alpha}
  \end{pmatrix} e^{\frac{\pi i \beta}{2} \sigma_3} E_2^{-1} \left (I-\frac{\kappa}{2\pi i}\ln \zeta ~\sigma_+\right )
= e^{\pi i \alpha} \begin{pmatrix}
    e^{\frac{\beta \pi i}{2} } & * \\
    e^{(\frac{3 \beta}{2} - 2 \alpha) \pi i } & *
  \end{pmatrix},
\end{equation}
whose first column is $e^{\pi i \alpha}$ multiplying the first column in the  matrix in \eqref{E2-1-product-matrix-1}. This implies that the asymptotics in \eqref{Phi0-limit-(-s)-pure I beta} should include an additional constant $e^{\pi i \alpha}$ when $2\alpha \in \mathbb{N}_0$, which also leads to \eqref{eq:asy-phi-0-negative infty-omega}.

This completes the proof of Proposition \ref{Pro: asy phi negative infty-omega}.
\end{proof}


\section{Proofs of main theorems} \label{Sec:Main-proof}

In this section, we will prove main results of this paper and begin with some useful differential identities for the integrals.

\subsection{Differential identities for the integrals}
As the asymptotics derived in Section \ref{sec:asymptotics} is for the function $w(s;2\alpha+\frac12,\omega)$ instead of for $q(s;2\alpha+\frac12,\omega)$, it is convenient for us to consider the following integral at this stage:
\begin{align}
  \widetilde{I}_1(s;\alpha,\omega) & :=\int_{c} ^{+\infty} \left( w(\tau;2\alpha+\frac{1}{2},\omega)+\sqrt{\tau}+\frac{\alpha+\frac{1}{4}}{\tau} \right) d\tau \nonumber \\
  &  \quad + \textrm{P.V.}\int_{s} ^{c}w(\tau;2\alpha+\frac{1}{2},\omega)d\tau +\frac 23 c^{\frac 32}+(\alpha+\frac{1}{4}) \ln c, \qquad c>0, \label{def:I12}
\end{align}
and introduce one more auxiliary integral
\begin{align}\label{def:I-3}
I_3(s;\alpha,\omega) &:= \int_{c}^{+\infty} \left( u(\tau;2\alpha,\omega)w'(\tau;2\alpha+\frac{1}{2},\omega)+H(\tau;2\alpha,\omega)+2\alpha\sqrt{\tau}+\frac{\alpha(2\alpha+1)}{2\tau} \right) d\tau \nonumber \\
&  \quad + \textrm{P.V.} \int_s^c \left( u(\tau; 2\alpha,\omega)w'(\tau;2\alpha+\frac{1}{2},\omega)+H(\tau;2\alpha,\omega) \right) d\tau
\nonumber
\\
& \quad +\frac{4}{3}\alpha c^{3/2}+\alpha(\alpha+\frac{1}{2})\ln c,
\end{align}
where $c$ is chosen such that the integrand over $(c, +\infty)$ is pole-free, and the Cauchy principle value is used if the integrand over $(s,c)$ has poles. In particular, it follows from \eqref{def:u} and \eqref{eq:Hamiltonian system} that $u(s; 2\alpha,\omega)w'(s;2\alpha+\frac{1}{2},\omega)$ only has simple poles on the real axis, which implies well-posedness of the Cauchy principal value in \eqref{def:I-3}. From the relation between $w$ and $q$ in \eqref{def:w}, it is easily seen that
\begin{equation} \label{eq: relation-I11}
  I_1(s;\alpha,\omega) = - \widetilde{I}_1(-2^{-\frac{1}{3}} s ;\alpha,\omega) - \frac{\alpha + \frac{1}{4}}{3} \ln2,
\end{equation}
where the integral $I_1$ is defined in \eqref{def:I11-new} and \eqref{def:I12-new}.

The aim of this section is to establish the following lemma which gives us the relations among the integrals $\widetilde{I}_1(s;\alpha,\omega)$,  $I_2(s;\alpha,\omega)$ and $I_3(s;\alpha,\omega)$.

\begin{lem}\label{lem: total integrals of w, v}
For $c>0$ and $\omega=e^{-2\beta \pi i} \geq 0$, the integrals $I_2$, $\widetilde I_1$  and $I_3$ defined in \eqref{def:I-2}, \eqref{def:I12} and \eqref{def:I-3} are convergent and independent of $c$. Moreover, we have the following identities
\begin{align}\label{eq:I-1-Psi}
\widetilde{I}_1(s;\alpha,\omega)& =-(2\alpha+\frac{1}{2})\ln2   -\ln|\left(\Phi_0\right)_{11}(s)|,
\\
\label{eq:I-2-deff}
\frac{\partial}{\partial\alpha}I_3(s;\alpha,\omega)& = -u(s) \frac{\partial}{\partial \alpha}w(s)+2\widetilde{I}_1(s;\alpha,\omega),
\\
\frac{\partial}{\partial\beta}I_3(s;\alpha,e^{-2\beta \pi i}) & =-u(s) \frac{\partial}{\partial \beta}w(s),\label{eq:I-2-deff2}
\end{align}
and
\begin{align}
  I_2(s;\alpha,\omega) & =-\frac{1}{3}\left(u(s;2\alpha,\omega)w(s;2\alpha+\frac{1}{2}, \omega)+2sH(s;2\alpha,\omega) +2\alpha^2 +\alpha \right) \nonumber \\
  & \quad +2\alpha \widetilde{I}_1(s; \alpha,\omega) -I_3(s;\alpha,\omega). \label{eq: integral H}
\end{align}
\end{lem}
\begin{proof}
The convergence of the relevant integrals over $(c,+\infty)$ in  \eqref{def:I12} and \eqref{def:I-3} follows directly from the large positive $s$ asymptotics of the functions $u(s;2\alpha,\omega)$, $w(s;2\alpha+\frac {1}{2},\omega)$ and $H(s;2\alpha,\omega)$ established in \eqref{eq:u-asy}--\eqref{eq:H-asy}. The existence of the Cauchy principle values therein follows from the fact that the integrants posses only simple poles on the real axis,  as  shown after \eqref{def:I12-new} and \eqref{def:I-3}.
Moreover, it is readily seen that the derivative of each integral with respect to $c$ is 0, which implies that they are all $c$-independent.

To show \eqref{eq:I-1-Psi}, we integrate on both sides of \eqref{eq:w-anti-derivative} from $s$ to $L$ and obtain after a straightforward calculation that
\begin{align}
 &\textrm{P.V.} \int_{s} ^{L}\frac{d}{d\tau}\ln\left(|\left(\Phi_0\right)_{11}(\tau)|\right)d\tau=\ln|\left(\Phi_0\right)_{11}(L)|-\ln|\left(\Phi_0\right)_{11}(s)|
\nonumber
\\
 & = \int_{c} ^{L} \left( w(\tau)+\sqrt{\tau}+\frac{\alpha+\frac{1}{4}}{\tau} \right) d\tau + \textrm{P.V.} \int_{s} ^{c}w(\tau)d\tau +\frac 23 c^{\frac 32}\nonumber \\
  & \quad+(\alpha+\frac{1}{4}) \ln c -\frac 23 L^{\frac 32}-(\alpha+\frac{1}{4}) \ln L,
\end{align}
where $c>0$ is large enough such the integrand $w(s)$ is pole-free on $(c,+\infty)$, which is possible in view of \eqref{eq:w-asy}. Since $w(s)$ only has simple poles on the real axis, the Cauchy principal value of the integral for $w(\tau)$ and thus  that of $\frac{d}{d\tau}\ln\left|\left(\Phi_0\right)_{11}(\tau)\right|$ on $[s,L]$ exists if $s,L\in \mathbb{R}$ are not poles of $w$.
By setting $L \to +\infty$, it follows from \eqref{eq:asy-phi-0} that
\begin{align}
  \widetilde{I}_1(s;\alpha,\omega) &= \lim_{L \to +\infty} \left( \ln|\left(\Phi_0\right)_{11}(L)| +\frac 23 L^{\frac 32}+(\alpha+\frac{1}{4}) \ln L \right)-\ln|\left(\Phi_0\right)_{11}(s)|
\nonumber
\\
&=-(2\alpha+\frac{1}{2})\ln2   -\ln|\left(\Phi_0\right)_{11}(s)|,
\end{align}
which is \eqref{eq:I-1-Psi}.

Next, by taking derivative with respect to $\alpha$ on both sides of \eqref{def:I-3}, it follows from the first equation in \eqref{eq:differential alpha} that
\begin{equation}
  \frac{\partial}{\partial\alpha}I_3(s;\alpha,\omega)= \lim_{\tau \to +\infty}\left(u(\tau) \frac{\partial}{\partial \alpha}w(\tau)\right) -u(s) \frac{\partial}{\partial \alpha}w(s)+2\widetilde{I}_1(s;\alpha,\omega).
\end{equation}
This, together with the large positive $s$ asymptotics of $u$ and $w$ given in \eqref{eq:u-asy} and \eqref{eq:w-asy}, gives us \eqref{eq:I-2-deff}. In a similar manner, one can prove \eqref{eq:I-2-deff2} with the aid of the second equation in \eqref{eq:differential alpha}, \eqref{eq:u-asy} and \eqref{eq:w-asy}.

Finally, we integrate on both sides of \eqref{eq:total differential} and obtain
\begin{align*}
  & u(L)w(L)+2LH(L)- u(s)w(s)-2sH(s) \\
  & = 3~ \textrm{P.V.}  \int_s^L (u(\tau)w'(\tau)+H(\tau)) d\tau + 3 ~ \textrm{P.V.}  \int_s^L (H(\tau)-2\alpha w(\tau)) d \tau.
\end{align*}
In view of the integrals $I_2$, $\widetilde I_1$  and $I_3$ defined in \eqref{def:I-2}, \eqref{def:I12} and \eqref{def:I-3}, respectively, we have, after letting $L \to +\infty$ in the above formula, that
\begin{align}
  & 3 I_3(s;\alpha,\omega) + 3 I_2(s;\alpha,\omega) - 6\alpha  \widetilde{I}_1(s;\alpha,\omega) \nonumber \\
  &  = \lim_{L \to +\infty}  \left(  u(L)w(L)+2LH(L) + 4\alpha L^{3/2} \right) - u(s)w(s)-2sH(s) \nonumber \\
  &=-\alpha-2\alpha^2- u(s)w(s)-2sH(s),
\end{align}
which is equivalent to  \eqref{eq: integral H}. In the second equality of the above formula, we have made use of the asymptotics established in Proposition \ref{Pro: asy large s}.

This completes the proof of Lemma \ref{lem: total integrals of w, v}.
\end{proof}

\subsection{Proof of Theorem \ref{thm:totalintPII}}

We have proved the convergence and $c$-independence of the integral $\widetilde{I}_1(s;\alpha,\omega)$ in Lemma \ref{lem: total integrals of w, v}. Moreover, from \eqref{eq:I-1-Psi} and the asymptotics of $\ln|\left(\Phi_0\right)_{11}(s)|$ given in \eqref{eq:asy-phi-0-negative infty} and \eqref{eq:asy-phi-0-negative infty-omega}, it follows that, as $s\to-\infty$,
\begin{align}
  \widetilde{I}_1(s;\alpha,0)& =
 -(2\alpha+\frac{1}{2})\ln |s|-\frac{1}{2}\ln (2\pi)+\ln \Gamma(1+2\alpha)+O(|s|^{-3/2}), \label{eq:total integral PII-old} \\
 \widetilde{I}_1(s;\alpha,e^{-2\beta \pi i}) & =
 -(\frac{\alpha}{2} - \frac{1}{4})\ln |s|-(3 \alpha+1) \ln 2  -\frac{\beta}{2} \pi i  + \ln \frac{\Gamma(1+2\alpha)}{|\Gamma(1+\alpha-\beta)|} \nonumber \\
 & \quad -\ln\left|\cos\left(\frac{\theta(s)}{2}+\arg\Gamma(1+\alpha-\beta)-\frac{\pi}{4}\right)\right| \label{eq:total integral PII-omega-old}
 +O(|s|^{-3/2}),
\end{align}
where $\theta(s)$ is defined in \eqref{theta}. A combination of the relation between $I_1$ and $\widetilde{I}_1$ in \eqref{eq: relation-I11} and the above two formulas then gives us Theorem \ref{thm:totalintPII}.

This completes the proof of Theorem \ref{thm:totalintPII}. \qed

\subsection{Proof of Theorem \ref{thm:Hamil}}

The convergence and $c$-independence of $I_2(s;\alpha,\omega)$ have already been proved in Lemma \ref{lem: total integrals of w, v}.
To show the asymptotic behavior of $I_2(s;\alpha,\omega)$ for large negative $s$, we observe from \eqref{eq: integral H}, \eqref{eq:total integral PII-old} and \eqref{eq:total integral PII-omega-old} that it suffices to derive the asymptotics of the auxiliary integral $I_3(s;\alpha,\omega)$ defined in \eqref{def:I-3}.

We begin with the proof of \eqref{eq:total integral H}. Inserting the asymptotics for $u(s;2\alpha,0)$, $w(s;2\alpha+\frac {1}{2}, 0)$ and $\widetilde{I}_1(s;\alpha,0)$ in \eqref{eq:asy-u-negative infty}, \eqref{eq:asy-w-negative infty} and \eqref{eq:total integral PII-old} into \eqref{eq:I-2-deff}, it is easily seen that
\begin{equation}\label{eq:I-3 diff}
\frac{\partial I_3(s;\alpha,0)}{\partial \alpha}= -2(2\alpha+\frac{1}{2})\ln |s|+2\ln \Gamma(1+2\alpha)+1-\ln (2\pi)+O(|s|^{-3/2}), \qquad s \to -\infty.
\end{equation}
Integrating the above formula with respect to $\alpha$ yields 
\begin{equation}\label{eq:I-3 }
I_3(s;\alpha,0)= I_3(s;0,0)-(2\alpha^2+\alpha)\ln|s|-(\ln(2\pi)-1)\alpha+2\int_0^{\alpha}\ln \Gamma(1+2\tau)d\tau+O(|s|^{-3/2}),
\end{equation}
as $s \to -\infty$. In view of \eqref{eq: integral H} with $\alpha =0$, we have
\begin{equation}\label{eq:I-3-inital }
I_3(s;0,0)=-\int_s^{\infty}H(\tau;0,0)d\tau-\frac{1}{3}\left(u(s;0,0)w(s;\frac{1}{2},0)+2sH(s;0,0)\right).
\end{equation} 
This, together with \eqref{eq:TW-large gap asy} and the asymptotics for $u(s;0,0)$ and $w(s;\frac{1}{2},0)$ in \eqref{eq:asy-u-negative infty} and \eqref{eq:asy-w-negative infty}, gives us
\begin{equation}
  I_3(s;0,0) = - \frac{s^3}{12} - \frac{1}{8} \ln |s| +c_0 + \frac{1}{6} + O(|s|^{-3/2}), \qquad s \to -\infty,
\end{equation}
where the constant $c_0$ is given in \eqref{def:c0}. For the integral involving the Gamma function in \eqref{eq:I-3 }, we recall the following relation between the Barnes' $G$-function and Loggamma integral
\begin{equation}\label{G-function}
  \ln G(1+x)=\frac{x}{2}\ln (2\pi)-\frac{x(x+1)}{2}+x\ln \Gamma(x+1)-\int_0^x \ln \Gamma(1+t)dt, \quad \mathrm{Re} ~x>-1;
\end{equation}
see \cite[Equation 5.17.4]{DLMF}. Thus, from \eqref{eq:I-3 } and the above two formulas, we have
\begin{align}
  I_3(s;\alpha,0) & = - \frac{s^3}{12} - ( 2 \alpha^2  +\alpha +  \frac{1}{8}) \ln |s| +c_0 + \frac{1}{6} \nonumber \\
   & \quad -2\alpha (\alpha + 1) + 2\alpha \ln \Gamma(1+2\alpha) - \ln G(1+2\alpha) +O(|s|^{-3/2}), \qquad s \to -\infty. \label{eq:I-3-alpha-asy}
\end{align}
The asymptotics of the integral for $H (\tau;2\alpha,0)$ in \eqref{eq:total integral H} then follows by substituting the asymptotics of $\widetilde{I}_1(s;\alpha,0)$ in \eqref{eq:total integral PII-old} and $I_3(s;\alpha,0)$ in \eqref{eq:I-3-alpha-asy} into the relation \eqref{eq: integral H}, where the asymptotics for $u(s;2\alpha,0)$, $w(s;2\alpha+\frac {1}{2}, 0)$ and $H(s;2\alpha,0)$ in \eqref{eq:asy-u-negative infty}, \eqref{eq:asy-w-negative infty} and \eqref{eq:asy-H-negative infty} are used as well.

The proof of the integral for $H(s; 2\alpha,e^{-2\beta \pi i})$ with $\beta i \in \mathbb{R}$ is a little bit involved. We first note from \eqref{u-special-solution} and \eqref{H-special-solution} that
 \begin{equation}\label{I-initial-beta}
  I_3(s;0,1) = 0.
\end{equation}
Thus,  we obtain from \eqref{eq:I-2-deff2} that
\begin{equation}\label{I3-initial-beta}
  I_3(s;0,e^{-2 \beta \pi i}) =- \int_0^{\beta} u(s;0,e^{-2 t \pi i}) \frac{\partial}{\partial t} w(s; \frac 12, e^{-2 t \pi i})dt.
\end{equation}
Substituting the asymptotics of $u(s)$ and $w(s)$ in \eqref{eq:asy-u-negative infty-omega} and \eqref{eq:asy-w-negative infty-omega} into the above formula,
we have, as $s\to-\infty$,
\begin{align}\label{I-3-alpha-0}
  & I_3(s;0,e^{-2 \beta \pi i})
  \nonumber
\\
  &= 2i \int_0^{\beta}t\left(i\ln\left(8|s|^{3/2}\right)-\frac{d}{dt}\arg(\Gamma(1-t)\right)dt+O\left(\frac{\ln |s|}{|s|^{3/2}}\right)\nonumber\\
  &= -\beta^2\ln\left(8|s|^{3/2}\right)-2i\beta\arg(1-\beta)+2i\int_0^{\beta}\arg\Gamma(1-t)dt+O\left(\frac{\ln |s|}{|s|^{3/2}}\right).
\end{align}
Similarly, from  the asymptotics of $u(s)$ and $w(s)$ in \eqref{eq:asy-u-negative infty-omega} and \eqref{eq:asy-w-negative infty-omega}, we have, as $s\to-\infty $,
 \begin{align}\label{uw-alpha}
 u(s)w_{\alpha}(s)
&= \frac{2|\alpha-\beta|\cos\left(\frac{\theta(s)}{2}+\arg\Gamma(\alpha-\beta)+\frac{\pi}{4}\right)}
  {\cos\left(\frac{\theta(s)}{2}+\arg\Gamma(1+\alpha-\beta)-\frac{\pi}{4}\right)}
\nonumber
\\
& \quad \times \frac{d}{d\alpha}\left(\frac{\theta(s)}{2}+\arg\Gamma(1+\alpha-\beta)-\frac{\pi}{4}\right)+O(|s|^{-3/2}) \nonumber\\
  &=\left( 2i\beta-2\alpha\tan\left(\frac{\theta(s)}{2}+\arg\Gamma(1+\alpha-\beta)-\frac{\pi}{4}\right)\right)
\nonumber
\\
& \quad \times
\frac{d}{d\alpha}\left(\frac{\theta(s)}{2}+\arg\Gamma(1+\alpha-\beta)-\frac{\pi}{4}\right)+O(|s|^{-3/2})\nonumber\\
  &=-\pi i\beta+2i\beta\frac{d}{d\alpha}\arg\Gamma(1+\alpha-\beta)
\nonumber
\\
& \quad
+2\alpha\frac{d}{d\alpha}\ln \left| \cos\left(\frac{\theta(s)}{2}+\arg\Gamma(1+\alpha-\beta)-\frac{\pi}{4}\right)\right|+O(|s|^{-3/2}).
\end{align}
This, together with \eqref{eq:I-2-deff} and \eqref{eq:total integral PII-omega-old}, gives us
 \begin{align}\label{I-3-alpha-omega-der}
 & \frac{\partial}{\partial \alpha}I_3(s;\alpha,e^{-2 \beta \pi i})=2\widetilde{I}_1(s;\alpha,e^{-2 \beta \pi i})-u(s)w_{\alpha}(s)\nonumber\\
  &=-(\alpha+\frac{1}{2})\ln|s|-(6\alpha+2)\ln2+2\ln\left(\frac{\Gamma(1+2\alpha)}{|\Gamma(1+\alpha-\beta)|}\right)-2i\beta\frac{d}{d\alpha}\arg\Gamma(1+\alpha-\beta) \nonumber \\
 &\quad -2\frac{d}{d\alpha}\left(\alpha\ln\left| \cos\left(\frac{\theta(s)}{2}+\arg\Gamma(1+\alpha-\beta)-\frac{\pi}{4}\right)\right|\right)+O(|s|^{-3/2}).
\end{align}
Integrating both sides of the above equation, we then obtain from the asymptotics of $I_3(s;0,e^{-2 \beta \pi i})$ in  \eqref{I-3-alpha-0} that
    \begin{align}\label{I-3-alpha-omega}
  &I_3(s;\alpha,e^{-2 \beta \pi i})
\nonumber
\\ &=  -\beta^2\ln\left(8|s|^{3/2}\right)-\frac{1}{2}\alpha(\alpha-1)\ln|s|-2\alpha\ln\left|\cos\left(\frac{\theta(s)}{2}+\arg\Gamma(1+\alpha-\beta)-\frac{\pi}{4}\right)\right|\nonumber\\
 &\quad  -(3\alpha+2)\alpha\ln 2-2i\beta\arg(1+\alpha-\beta)+2i\int_0^{\beta}\arg\Gamma(1-t)dt\nonumber\\
 &\quad +2\int_0^{\alpha}\ln \Gamma(1+2t)dt -2\int_0^{\alpha}\ln|\Gamma(1-\beta+t)|dt+O\left(\frac{\ln |s|}{|s|^{3/2}}\right).
\end{align}
Next, we have from the asymptotics of $u,w,H$ in \eqref{eq:asy-u-negative infty-omega}, \eqref{eq:asy-w-negative infty-omega} and \eqref{eq:asy-H-negative infty-omega} that, as $s\to-\infty $,
\begin{equation}\label{uw-sH}
u(s)w(s)+2sH(s)+2\alpha^2+\alpha=-4i\beta|s|^{3/2}+3(\alpha^2-\beta^2)+O(|s|^{-3/2}).
\end{equation}
Finally, substituting \eqref{eq:total integral PII-omega-old}, \eqref{I-3-alpha-omega}  and \eqref{uw-sH} into \eqref{eq: integral H}, it follows
\begin{align}\label{I-2-alpha-omega}
  I_2(s;\alpha,e^{-2 \beta \pi i}) &= \frac{4}{3}i\beta |s|^{3/2}+(\beta^2-\frac{1}{3}\alpha^2)\ln\left(8|s|^{3/2}\right)-2\alpha^2\ln2  +(\beta^2-\alpha^2)\nonumber\\
  &\quad +2\alpha\ln\Gamma(1+2\alpha)-2\alpha\ln|\Gamma(1+\alpha-\beta)|+2i\beta\arg(1+\alpha-\beta) \nonumber \\
& \quad -2\int_0^{\alpha}\ln \Gamma(1+2t)dt +2\int_0^{\alpha}\ln|\Gamma(1-\beta+t)|dt
\nonumber\\
 & \quad -2i\int_0^{\beta}\arg\Gamma(1-t)dt+O\left(\frac{\ln |s|}{|s|^{3/2}}\right).
\end{align}
Since $i\beta \in \mathbb{R}$, it follows that
\begin{align}\label{I-2-1}
  2\int_0^{\alpha}\ln|\Gamma(1-\beta+t)|dt&=\int_0^{\alpha}\ln\Gamma(1+\beta+t)+\ln\Gamma(1-\beta+t)dt
\nonumber\\
  &=\int_{\beta}^{\alpha+\beta}\ln\Gamma(1+t)dt+\int_{-\beta}^{\alpha-\beta}\ln\Gamma(1+t)dt,
\end{align}
and
\begin{align}\label{I-2-2}
  -2i\int_0^{\beta}\arg\Gamma(1-t)dt &=\int_0^{\beta}\ln\Gamma(1+t)-\ln\Gamma(1-t)dt
\nonumber\\
  &=\int_{0}^{\beta}\ln\Gamma(1+t)dt+\int_{0}^{-\beta}\ln\Gamma(1+t)dt.
\end{align}
Substituting \eqref{I-2-1} and \eqref{I-2-2} into \eqref{I-2-alpha-omega}, we then obtain \eqref{thm: H-integral} from the definition of Barnes' $G$-function given in \eqref{G-function}.

This completes the proof of Theorem \ref{thm:Hamil}. \qed

\begin{appendices}
\section{Airy, Bessel and confluent hypergeometric parametrices}

\subsection{Airy parametrix}\label{sec:Airy}
Define
 \begin{equation}\label{Phi-A-solution}
  \Phi^{(\Ai)}(z)=M\left\{
                 \begin{array}{ll}
                  \begin{pmatrix}
                        \Ai(z) & \Ai(e^{-\frac{2\pi i}{3}}  z) \\
                   \Ai'(z)& e^{-\frac{2\pi i}{3}} \Ai'(e^{-\frac{2\pi i}{3}} z) \\
                      \end{pmatrix}e^{-\frac{\pi i}{6}\sigma_3}, &z\in \texttt{I}, \\[.4cm]
                      \begin{pmatrix}
                        \Ai(z) & \Ai(e^{-\frac{2\pi i}{3}} z) \\
                      \Ai'(z)& e^{-\frac{2\pi i}{3}} \Ai'(e^{-\frac{2\pi i}{3}} z) \\
                     \end{pmatrix}e^{-\frac{\pi i}{6}\sigma_3}\begin{pmatrix}
                                                    1 & 0 \\
                                                        -1 & 1 \\
                                                     \end{pmatrix}
                 , & z\in \texttt{II}, \\[.4cm]
                            \begin{pmatrix}
                         \Ai(z) & -e^{-\frac{2\pi i}{3}} \Ai(e^{\frac{2\pi i}{3}} z) \\
                        \Ai'(z)&- \Ai'(e^{\frac{2\pi i}{3}} z) \\
                  \end{pmatrix}e^{-\frac{\pi i}{6}\sigma_3}
                  \begin{pmatrix}
                                                       1 & 0 \\
                                                      1 & 1 \\
                                                      \end{pmatrix}
                 , &z\in \texttt{III},
                 \\[.4cm]
                      \begin{pmatrix}
                    \Ai(z) & - e^{-\frac{2\pi i}{3}} \Ai(e^{\frac{2\pi i}{3}} z) \\
                        \Ai'(z)&- \Ai'(e^{\frac{2\pi i}{3}} z) \\
                     \end{pmatrix}e^{-\frac{\pi i}{6}\sigma_3}, & z\in \texttt{IV},
                 \end{array} \right.
   \end{equation}
where $\Ai(z)$ is the standard Airy function (cf. \cite[Chapter 9]{DLMF}),
$$M=\sqrt{2\pi} e^{\frac 1 6\pi i}
\begin{pmatrix}
1 & 0 \\
0 & -i
\end{pmatrix},
$$
and the regions $\texttt{I-IV}$ are indicated in Fig. \ref{fig:jumps-Phi-A}.
It is well-known that $\Phi^{(\Ai)}(z)$ is the solution of the following RH problem; cf.\cite{D}.

 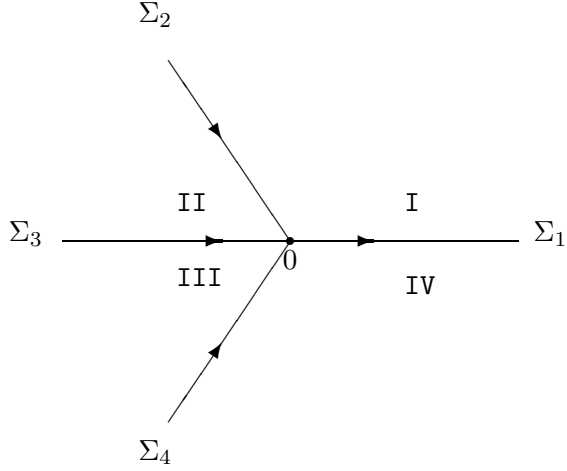
\begin{figure}[t]
\begin{center}
   \setlength{\unitlength}{1truemm}
   \begin{picture}(100,70)(-5,2)
       \put(40,40){\line(-2,-3){16}}
       \put(40,40){\line(-2,3){16}}
       \put(40,40){\line(-1,0){30}}
       \put(40,40){\line(1,0){30}}

       \put(30,55){\thicklines\vector(2,-3){1}}
       \put(30,40){\thicklines\vector(1,0){1}}
       \put(50,40){\thicklines\vector(1,0){1}}
       \put(30,25){\thicklines\vector(2,3){1}}

       \put(39,36.3){$0$}
       \put(20,11){$\Sigma_4$}
       \put(20,69){$\Sigma_2$}
       \put(3,40){$\Sigma_3$}

       \put(72,40){$\Sigma_1$}
       \put(25,44){$\texttt{II}$}
       \put(25,34){$\texttt{III}$}
       \put(55,44){$\texttt{I}$}
       \put(55,33){$\texttt{IV}$}

       \put(40,40){\thicklines\circle*{1}}

   \end{picture}
   \caption{The jump contours  and  regions for the RH problem for $\Phi^{(\Ai)}$.}
   \label{fig:jumps-Phi-A}
\end{center}
\end{figure}

\subsection*{RH problem for $\Phi^{(\Ai)}$}

\begin{description}
  \item(a)   $ \Phi^{(\Ai)}(z)$ is analytic in
  $\mathbb{C}\setminus \{\cup^4_{j=1}\Sigma_j\cup\{0\}\}$, where the contours $\Sigma_j$, $j=1,2,3,4$, are illustrated in Fig.  \ref{fig:jumps-Phi-A}.

  \item(b)   $ \Phi^{(\Ai)}(z)$  satisfies the jump condition
\begin{equation}
   \Phi^{(\Ai)}_+(z)= \Phi^{(\Ai)}_{-}(z) \widetilde J_{i}(z), \quad  z\in \Sigma_i, \quad i=1,2,3,4,
   \end{equation}
where
$$
\widetilde J_{1}(z)=\begin{pmatrix}
1 & 1 \\
0 & 1
\end{pmatrix}, \quad
\widetilde J_{2}(z) = \widetilde J_{4}(z)=\begin{pmatrix}
1 & 0 \\
1 & 1
\end{pmatrix},
\quad
\widetilde J_{3}(z)=
\begin{pmatrix}
0 & 1 \\
-1 & 0
\end{pmatrix}.
$$

 \item(c)   $ \Phi^{(\Ai)}(z)$ satisfies the following asymptotic behavior  at infinity:
  \begin{equation}\label{phi at infinity}   \Phi^{(\Ai)}(z)=z^{-\frac{1}{4}\sigma_3}\frac{I+i\sigma_1}{\sqrt{2}}
   \left (I+O\left (\frac 1{z^{3/2}}\right )\right)e^{-\frac{2}{3} z^{3/2}\sigma_3}, \qquad  z \rightarrow \infty,
 \end{equation}
 where the matrices $\sigma_1$ and $\sigma_3$ are given in \eqref{def:Paulimatrice}.
\end{description}

\subsection{Bessel parametrix }\label{sec:Bessel}
Define
\begin{equation}\label{Phi-B-solution}
\Phi^{(\Bes)}_{\alpha}(z)=\left\{
                             \begin{array}{ll}
                               \begin{pmatrix}
I_{\alpha}(z^{1/2}) & \frac{i}{\pi}K_{\alpha}(z^{1/2}) \\
\pi iz^{1/2}I'_{\alpha}(z^{1/2}) &
-z^{1/2}K_{\alpha}'(z^{1/2})
\end{pmatrix}, & z\in \texttt{I},\\
                              \begin{pmatrix}
I_{\alpha}(z^{1/2}) & \frac{i}{\pi}K_{\alpha}(z^{1/2}) \\
\pi iz^{1/2}I'_{\alpha}(2z^{1/2}) &
-z^{1/2}K_{\alpha}'(z^{1/2})
\end{pmatrix}\begin{pmatrix}
                                1 & 0\\
                               -e^{\alpha \pi i} & 1
                                \end{pmatrix}, & z\in \texttt{II}, \\
                                \begin{pmatrix}
I_{\alpha}(z^{1/2}) & \frac{i}{\pi}K_{\alpha}(z^{1/2}) \\
\pi iz^{1/2}I'_{\alpha}(z^{1/2}) &
-z^{1/2}K_{\alpha}'(z^{1/2})
\end{pmatrix}\begin{pmatrix}
1 & 0 \\
e^{-\alpha\pi i} & 1
\end{pmatrix}, &  z\in \texttt{III},
                             \end{array}
                           \right.
\end{equation}
where  $I_\alpha(z)$ and $K_\alpha(z)$ denote the  modified Bessel functions (cf. \cite[Chapter 10]{DLMF}), the principle branch is taken for $z^{1/2}$ and the regions $\texttt{I-III}$ are illustrated in Fig. \ref{fig:jumps-Phi-B}. By \cite{KMVV},  we have that $\Phi^{(\Bes)}_{\alpha}(z)$ satisfies the RH problem below.

   \begin{figure}[t]
\begin{center}
   \setlength{\unitlength}{1truemm}
   \begin{picture}(80,70)(-5,2)
       \put(40,40){\line(-2,-3){16}}
       \put(40,40){\line(-2,3){16}}
       \put(40,40){\line(-1,0){30}}

       \put(30,55){\thicklines\vector(2,-3){1}}
       \put(30,40){\thicklines\vector(1,0){1}}
       \put(30,25){\thicklines\vector(2,3){1}}

       \put(39,36.3){$0$}
       \put(20,11){$\Sigma_3$}
       \put(20,69){$\Sigma_1$}
       \put(3,40){$\Sigma_2$}

       \put(52,39){$\texttt{I}$}
       \put(25,44){$\texttt{II}$}
       \put(25,34){$\texttt{III}$}

       \put(40,40){\thicklines\circle*{1}}

   \end{picture}
   \caption{The jump contours and regions for the RH problem for $\Phi^{(\Bes)}_{\alpha}$.}
   \label{fig:jumps-Phi-B}
\end{center}
\end{figure}
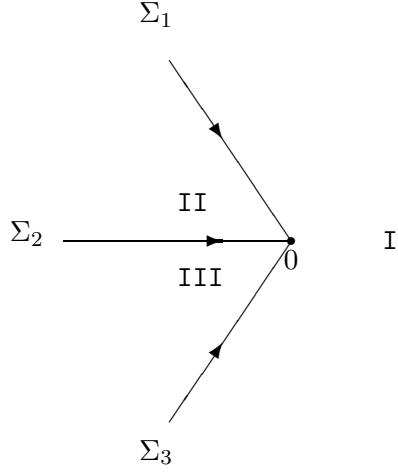

\subsection*{RH problem for $\Phi_\alpha^{(\Bes)}$}
\begin{description}
\item(a) $\Phi^{(\Bes)}_{\alpha}(z)$ is defined and analytic in $\mathbb{C}\setminus \{\cup^3_{j=1}\Sigma_j\cup\{0\}\}$, where the contours $\Sigma_j$, $j=1,2,3$,  are indicated  in Figure \ref{fig:jumps-Phi-B}.

\item(b) $\Phi^{(\Bes)}_{\alpha}(z)$ satisfies the following jump conditions:
\begin{equation}
 \Phi^{(\Bes)}_{\alpha,+}(z)=\Phi^{(\Bes)}_{\alpha,-}(z)
 \left\{
 \begin{array}{ll}
   \begin{pmatrix}
                                1 & 0\\
                               e^{\alpha \pi i} & 1
                                \end{pmatrix},  &  z \in \Sigma_1, \\
   \begin{pmatrix}
                                0 & 1\\
                               -1 & 0
                                \end{pmatrix},  &  z \in \Sigma_2, \\
   \begin{pmatrix}
                                 1 & 0 \\
                                 e^{-\alpha\pi i} & 1 \\
                                \end{pmatrix}, &   z \in \Sigma_3.
 \end{array}  \right .  \end{equation}

\item(c) $\Phi^{(\Bes)}_{\alpha}(z)$ satisfies the following asymptotic behavior at infinity:
\begin{multline}\label{eq:Besl-infty}
 \Phi^{(\Bes)}_{\alpha}(z)=
 ( \pi^2 z )^{-\frac{1}{4} \sigma_3} \frac{I + i \sigma_1}{\sqrt{2}}
 \\
 \times  \left( I + \frac{1}{8\sqrt{z}}  \begin{pmatrix}
   -1-4\alpha^2 & -2i \\
   -2i & 1+4\alpha^2
 \end{pmatrix} +  O\left(\frac{1}{z}\right)
 \right)e^{\sqrt{z}\sigma_3},\quad z\to \infty.
   \end{multline}

\item(d) $\Phi^{(\Bes)}_{\alpha}(z)$ satisfies the  following asymptotic behaviors near the origin:
\newline
    If $\alpha<0$,
    \begin{equation}
\Phi^{(\Bes)}_{\alpha}(z)=
O \begin{pmatrix}
|z|^{\alpha/2} & |z|^{\alpha/2}
\\
|z|^{\alpha/2} & |z|^{\alpha/2}
\end{pmatrix}, \qquad \textrm{as $z \to 0$}.
\end{equation}
If $\alpha=0$,
    \begin{equation}
\Phi^{(\Bes)}_{\alpha}(z)=
O \begin{pmatrix}
\ln|z| & \ln|z|
\\
\ln|z| & \ln|z|
\end{pmatrix}, \qquad \textrm{as $z \to 0$}.
\end{equation}
If $\alpha>0$,
     \begin{equation}
\Phi^{(\Bes)}_{\alpha}(z)= \begin{cases}
  O\begin{pmatrix}
|z|^{\alpha/2} & |z|^{-\alpha/2}
\\
|z|^{\alpha/2} & |z|^{-\alpha/2}
\end{pmatrix}, & \textrm{as $z \to 0$ and $z\in \texttt{I}$, } \\[.4cm]
  O\begin{pmatrix}
|z|^{-\alpha/2} & |z|^{-\alpha/2}
\\
|z|^{-\alpha/2} & |z|^{-\alpha/2}
\end{pmatrix}, & \textrm{as $z \to 0$ and $z\in \texttt{II}\cup \texttt{III} $. }
  \end{cases}
\end{equation}
\end{description}

\subsection{Confluent hypergeometric parametrix}\label{sec:CHF}

\begin{figure}[t]
\begin{center}
   \setlength{\unitlength}{1truemm}
   \begin{picture}(100,70)(-5,2)
       \put(40,40){\line(-2,-3){18}}
       \put(40,40){\line(-2,3){18}}
       \put(40,40){\line(-1,0){30}}
       \put(40,40){\line(1,0){30}}
      \put(40,40){\line(0,1){30}}
 \put(40,40){\line(0,-1){30}}
  \put(40,40){\line(2,-3){18}}
    \put(40,40){\line(2,3){18}}

       \put(30,55){\thicklines\vector(2,-3){1}}
       \put(30,40){\thicklines\vector(1,0){1}}
       \put(50,40){\thicklines\vector(1,0){1}}
       \put(30,25){\thicklines\vector(2,3){1}}
      \put(50,25){\thicklines\vector(2,-3){1}}
       \put(50,55){\thicklines\vector(2,3){1}}

     \put(40,58){\thicklines\vector(0,1){1}}
      \put(40,22){\thicklines\vector(0,-1){1}}

       \put(42,36.9){$0$}
         \put(72,40){$\Sigma_1$}
           \put(60,69){$\Sigma_2$}
             \put(39,73){$\Sigma_3$}
             \put(20,69){$\Sigma_4$}
              \put(3,40){$\Sigma_5$}
        \put(18,10){$\Sigma_6$}
               \put(39,5){$\Sigma_7$}
          \put(60,11){$\Sigma_8$}

          \put(55,48){$\texttt{I}$}
           \put(44,60){$\texttt{II}$}
              \put(31,60){$\texttt{III}$}
            \put(22,48){$\texttt{IV}$}
        \put(24,31){$\texttt{V}$}
        \put(31,19){$\texttt{VI}$}
         \put(44,19){$\texttt{VII}$}
 \put(52,31){$\texttt{VIII}$}

       \put(40,40){\thicklines\circle*{1}}
\end{picture}
   \caption{The jump contours  and  regions for the RH problem for $\Phi_{\alpha,\beta}^{(\CHF)}$.}
   \label{fig:jumps-Phi-C}
\end{center}
\end{figure}
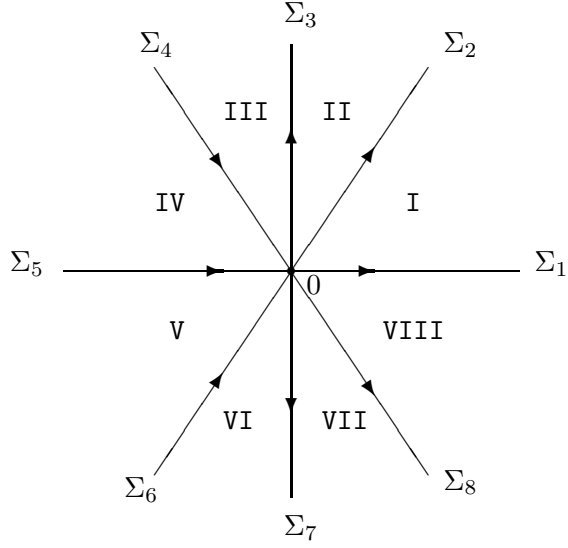

We start with the  following RH problem.
 \subsection*{RH problem for $\Phi_{\alpha,\beta}^{(\CHF)}$}
 \begin{description}
  \item(a)   $\Phi_{\alpha,\beta}^{(\CHF)}(z)$ is analytic in
  $\mathbb{C}\setminus \{\cup^8_{j=1}\Sigma_j\cup\{0\}\}$, where the contours $\Sigma_j$, $j=1,\ldots,8,$ are indicated in Fig. \ref{fig:jumps-Phi-C}.

  \item(b) $\Phi_{\alpha,\beta}^{(\CHF)}(z)$ satisfies the following jump condition:
  \begin{equation}\label{jumps-phi-c}
   \left(\Phi_{\alpha,\beta}^{(\CHF)}\right)_+(z)=\left(\Phi_{\alpha,\beta}^{(\CHF)}\right)_{-}(z) \widehat J_i(z), \quad z \in \Sigma_i,\quad j=1,\ldots,8,
  \end{equation}
  where
  \begin{equation*}
    \widehat J_1(z) = \begin{pmatrix}
    0 &   e^{-\pi i\beta} \\
    -  e^{\pi i\beta} &  0
    \end{pmatrix}, \quad \widehat J_2(z) = \begin{pmatrix}
    1 & 0 \\
    e^{ \pi i(\beta-2\alpha)} & 1
    \end{pmatrix},
                                                           \quad \widehat J_3(z) = \widehat J_7(z) = e^{\pi i\alpha\sigma_3},
  \end{equation*}
  \begin{equation*}
    \widehat J_4(z) = \begin{pmatrix}
    1 & 0 \\
    e^{ -\pi i(\beta-2\alpha)} & 1
    \end{pmatrix}, \quad
    \widehat J_5(z) = \begin{pmatrix}
    0 &   e^{\pi i\beta} \\
     -  e^{-\pi i\beta} &  0
     \end{pmatrix}, \quad
     \widehat J_6(z) = \begin{pmatrix}
     1 & 0 \\
     e^{- \pi i(\beta+2\alpha)} & 1
     \end{pmatrix},
  \end{equation*}
  and $$\widehat J_8(z) = \begin{pmatrix}
   1 & 0 \\
   e^{\pi i(\beta+2\alpha)} & 1
   \end{pmatrix}.
   $$

  \item(c) $\Phi_{\alpha,\beta}^{(\CHF)}(z)$ satisfies the following asymptotic behavior at infinity:
  \begin{multline}\label{B-tiled at infinity}
 \Phi_{\alpha,\beta}^{(\CHF)}(z) =  (I + O(1/z)) z^{-\beta \sigma_3}e^{-iz\sigma_3}
\\
 \times
  \left\{\begin{array}{ll}
                         I, & ~0< \arg z <\frac{ \pi}{2}, \\
                         e^{\pi i \alpha\sigma_3}, &~ \frac{\pi}{2}< \arg z<\pi,
                         \\
                        \begin{pmatrix}
                                                             0 &   -e^{\pi i(\alpha+\beta)} \\
                                                            e^{-\pi i(\alpha+\beta)} &  0
                        \end{pmatrix}, &~ \pi< \arg z<\frac{3\pi}{2},
                        \\
                        \begin{pmatrix}
                        0 &   -e^{-\pi i\beta} \\
                        e^{\pi i\beta} &  0
                        \end{pmatrix}, & -\frac{\pi}{2}<\arg z<0.
 \end{array}\right.
\end{multline}
\item(d) $\Phi_{\alpha,\beta}^{(\CHF)}(z)$ satisfies the following asymptotic behavior near the origin. As $z \to 0$,
$$\Phi_{\alpha,\beta}^{(\CHF)}(z)=\begin{pmatrix}
                         O(|z|^{\alpha}) & O(|z|^{-|\alpha|}) \\
                         O(|z|^{\alpha}) & O(|z|^{-|\alpha|}) \end{pmatrix}, \quad \alpha\neq 0,$$
and $$\Phi_{\alpha,\beta}^{(\CHF)}(z)=\begin{pmatrix}
                         O(1) & O(\ln |z|) \\
                         O(1) & O(\ln |z|)
\end{pmatrix}, \quad \alpha=0.$$
\end{description}

From \cite{DIK2011,ikj2008}, it follows that the above RH problem can be solved explicitly in the following way. For $z\in \texttt{I}$,
\begin{align}\label{Phi-C-solution}
\Phi_{\alpha,\beta}^{(\CHF)}(z)&=C_0\left(\begin{array}{ll}
(2e^{\frac{\pi i}{2}}z)^\alpha\psi(\alpha+\beta,1+2\alpha,2e^{\frac{\pi i}{2}}z)e^{i\pi(\alpha+2\beta)}e^{-iz}\\
-\frac{\Gamma(1+\alpha+\beta)}{\Gamma(\alpha-\beta)}(2e^{\frac{\pi i}{2}}z)^{-\alpha}\psi(1-\alpha+\beta,1-2\alpha,2e^{\frac{\pi i}{2}}z)e^{i\pi(-3\alpha+\beta)}e^{-iz}\end{array}\right.\nonumber\\
&\quad\quad\left.\begin{array}{rr}
-\frac{\Gamma(1+\alpha-\beta)}{\Gamma(\alpha+\beta)}(2e^{\frac{\pi i}{2}}z)^\alpha\psi(1+\alpha-\beta,1+2\alpha,2e^{-\frac{\pi i}{2}}\zeta)e^{i\pi(\alpha+\beta)}e^{iz}\\
(2e^{\frac{\pi i}{2}}z)^{-\alpha}\psi(-\alpha-\beta,1-2\alpha,2e^{-\frac{\pi i}{2}}z)e^{-i\pi\alpha}e^{iz}\end{array}\right),
\end{align}
where  the confluent hypergeometric function $\psi(a,b;z)$   is the unique solution to the Kummer's equation
$$z\frac{d^2y}{dz^2}+(b-z)\frac{dy}{dz}-ay=0$$
satisfying the boundary condition $\psi(a,b,z)\sim  z^{-a}$ as $z\to \infty$ and $-\frac{3\pi }{2} < \arg z < \frac{3\pi}{2}$;
see \cite[Chapter 13]{DLMF}.  The branches of the multi-valued functions are chosen such that $-\frac{\pi}{2}<\arg z<\frac{3\pi}{2}$ and
$$
 C_0=2^{\beta \sigma_3}e^{\beta\pi i\sigma_3/2}
\begin{pmatrix}
e^{-i\pi(\alpha+2\beta)} & 0
\\
0 & e^{i\pi(2\alpha+\beta)}
\end{pmatrix}
$$
is a constant matrix. The explicit formula of $\Phi_{\alpha,\beta}^{(\CHF)}$ in the other sectors is then determined by using the jump condition \eqref{jumps-phi-c}.
\end{appendices}

\section*{Acknowledgements}
Dan Dai was was partially supported by grants from the City University of Hong Kong (Project No. 7004864, 7005032), and grants from the Research Grants Council of the Hong Kong Special Administrative Region, China (Project No. CityU 11300115, CityU 11303016). Shuai-Xia Xu was partially supported by National Natural Science Foundation of China under grant numbers 11971492, 11571376 and 11201493. Lun Zhang was partially supported by National Natural Science Foundation of China under grant number 11822104, by The Program for Professor of Special Appointment (Eastern Scholar) at Shanghai Institutions of Higher Learning, and by Grant EZH1411513 from Fudan University.

\end{document}